\documentclass[12pt,journal,onecolumn]{IEEEtran}
\usepackage{amsthm}
\usepackage{amssymb}
\usepackage{amsmath}
\usepackage{algorithmic}
\usepackage{algorithm}
\usepackage{array}
\usepackage{arydshln}
\usepackage{nameref}
\newtheorem{theorem}{Theorem}
\newtheorem{lemma}[theorem]{Lemma}

\newtheorem{example}[theorem]{Example}
\allowdisplaybreaks
\usepackage[justification=centering]{caption}

\usepackage[para]{threeparttable}

\newtheorem{remark}[theorem]{Remark}

\usepackage{lineno}
\usepackage{blindtext}
\usepackage{threeparttable}
 \usepackage{graphicx}
 \usepackage{arydshln}
\allowdisplaybreaks[4]

\makeatletter

\newcommand{\Rmnum}[1]{\expandafter\@slowromancap\romannumeral #1@}
\makeatother

\ifCLASSINFOpdf

\else

\fi

\begin{document}

\title{Binary Cyclic Codes With Simultaneously Large Minimum Distances and Dual Distances}

\author{Ziling Heng and Jiantao Hu
\thanks{Z. Heng and J. Hu are with the School of Science, Chang’an University, Xi’an 710064, China, and also with the State Key Laboratory of
 Integrated Services Networks, Xidian University, Xi’an 710071, China  
 (Email: zilingheng@chd.edu.cn, 15320266040@163.com). (\emph{Corresponding author: Jiantao Hu})}
 \thanks{ Z. Heng’s research was supported by the National Natural Science
 Foundation of China under Grant 12271059.}
 }



\maketitle

\begin{abstract}
Binary cyclic codes are an important class of linear codes that admit highly efficient encoding and decoding algorithms.
Constructing binary cyclic codes with favorable parameters has been an active research topic for several decades.
In recent years, substantial progress has been made in the study of binary $[n,k,d]$ cyclic codes with $k$ close to $n/2$ and $d\geq \sqrt{n}$. 
Nevertheless, among the known codes of this type, only a few families simultaneously possess both large minimum distances and large dual distances.
In this paper, for even $s$, we focus on new constructions of $[2^{s}-1,k]$ binary cyclic codes with rate approximately $1/2$,
for which both the minimum distance $d$ and the dual distance $d^\perp$ are simultaneously large.
 Some of the constructions yield binary cyclic codes with parameters $[2^s-1,2^{s-1}\pm 1,d]$ such that the lower bounds of $d\cdot d^\perp$ are close to $n$ or $2n$.  The other constructions yield binary cyclic codes with parameters $[2^s -1, 2^{s-1} + c, d]$ for some $ 1- \frac{5s}{2} \leq c \leq \frac{5s}{2} -1$ such
that the lower bounds of $d\cdot d^\perp$ are close to $n$ or $2n$. Compared with the known binary cyclic codes of the same length and dimension, the majority of the binary cyclic codes we construct exhibit larger minimum distances. In particular, we obtain several codes with excellent parameters, which are verified against the Code Tables available at http://www.codetables.de/.
\end{abstract}

\begin{IEEEkeywords}
Binary code, cyclic code, minimum distance.
\end{IEEEkeywords}

\section{Introduction}
\IEEEPARstart{T}{hroughout} this paper, let $\mathbb{F}_q$ be the finite field with $q$ elements, where $q$ is a power of a prime $p$. An $[n, k, d]$ linear code $\mathcal{C}$ over $\mathbb{F}_q$ is defined by a $k$-dimensional $\mathbb{F}_q$-linear subspace of $\mathbb{F}_q^n$ with minimum (Hamming) distance $d$. The dual code of $\mathcal{C}$ is defined as
\[
\mathcal{C}^{\perp} = \{\mathbf{b} \in \mathbb{F}_q^n : <\mathbf{b}, \mathbf{c}> = 0 \text{ for all } \mathbf{c} \in \mathcal{C}\},
\]
where $<\cdot, \cdot>$ denotes the standard inner product. 
Then $\mathcal{C}^\perp$ is an $[n,n-k,d^\perp]$ linear code over $\mathbb{F}_q$, where $d^\perp$ is also called the dual distance of $\mathcal{C}$. 

A linear code $\mathcal{C}$ over $\mathbb{F}_q$ is said to be cyclic if any $(c_0, c_1, \ldots, c_{n-1}) \in \mathcal{C}$ implies $(c_{n-1}, c_0, \ldots, c_{n-2}) \in \mathcal{C}$. By identifying each vector $(c_0, c_1, \ldots, c_{n-1}) \in \mathbb{F}_q^n$ with
\[
c_0 + c_1 x + c_2 x^2 + \cdots + c_{n-1} x^{n-1} \in \mathbb{F}_q[x] / \langle x^n - 1 \rangle,
\]
a linear code $\mathcal{C}$ of length $n$ over $\mathbb{F}_q$ corresponds to a subset of the quotient ring $\mathbb{F}_q[x] / \langle x^n - 1 \rangle$. The code $\mathcal{C}$ is  cyclic  if and only if the corresponding subset is an ideal of $\mathbb{F}_q[x] / \langle x^n - 1 \rangle$. Since every ideal of $\mathbb{F}_q[x] / \langle x^n - 1 \rangle$ is principal, there exists a monic polynomial $g(x)$ of the smallest degree such that $\mathcal{C} = \langle g(x) \rangle$ with $g(x) \mid (x^n - 1)$. Then $g(x)$ and $h(x): = (x^n - 1)/g(x)$ are respectively referred to as the generator polynomial and check polynomial of $\mathcal{C}$. Moreover, the dual code $\mathcal{C}^\perp$ is generated by the reciprocal polynomial of $h(x)$.
If $\gcd(n,q)=1$, then $g(x)$ has no repeated roots in the algebraic closure of $\mathbb{F}_q$  and $\mathcal{C}$ is said to be a simple-root cyclic code. 
In this paper, we only study simple-root binary cyclic codes.

Cyclic codes have been widely studied due to their efficient encoding and decoding algorithms. Much progress has been made on this type of codes over the past 70 years (e.g. \cite{DING C, LDX, LIU and Ding, Liu2, wang and tang}).
It is known that binary quadratic-residue codes and the punctured Reed-Muller codes have parameters $[n,(n+1)/2,d\geq \sqrt{n}]$. 
Their rates are about $1/2$ and  minimum distances have a square-root bound. Such codes are interesting as they have both high code rate and strong error-correcting capability‌.  
This motivated scholars to construct new binary cyclic codes with code rates near $1/2$ and large minimum distances. 
As pointed out in \cite{CTCD}, constructing such codes is not easy. To the best of our knowledge, the following summarizes the known infinite families of such codes:
\begin{itemize}
    \item Binary $[n, (n+1)/2, d]$ quadratic residue codes and their even-weight subcodes were given in \cite[Section 6.6]{WCH}, where $d^2 \geq n$ and $n \equiv \pm 1 \pmod{8}$ is a prime. 
    
    \item The punctured binary $[2^s-1, 2^{s-1}, 2^{(s+1)/2}-1]$  Reed-Muller codes of order $(s-1)/2$ were studied in \cite{EFA},  where $s$ is odd and $d^\perp \geq 2^{(s+1)/2}$. 

    \item A family of binary $[2^s - 1, 2^{s-1}, d]$ cyclic codes for odd $s \geq 3$ was constructed in \cite{CTCD}, where $d \geq 2^{(s-1)/2} + 1$ and $d^{\perp} \geq 2^{(s-1)/2} + 2$ for $s \equiv 3 \pmod{4}$, and $d \geq 2^{(s-1)/2} + 3$ and $d^{\perp} \geq 2^{(s-1)/2} + 4$ for $s \equiv 1 \pmod{4}$. 

    \item A family of binary $[2^s - 1, 2^{s-1} - 1, d]$ cyclic codes for even $s \geq 4$ was obtained in  \cite{CTCD}, where $d \geq 2^{(s-2)/2} + 1$ and $d^{\perp} \geq 2^{(s-2)/2} + 2$ for $s \equiv 0 \pmod{4}$, and $d \geq 2^{(s-4)/2} + 1$ and $d^{\perp} \geq 2^{(s-4)/2} + 2$ for $s \equiv 2 \pmod{4}$. 
        
    \item A family of binary $[2^s - 1, 2^{s-1}, d]$ cyclic codes for $s \equiv 1 \pmod{4}$ was constructed in \cite{LIU}, where $d \geq 2^{(s-1)/2} + 3$ and $d^\perp \geq 2^{(s-1)/2} + 4$ if $s \equiv 1 \pmod{8}$, and $d \geq 2^{(s-1)/2} + 1$ and $d^\perp \geq 2^{(s-1)/2} + 2$  if $s \equiv 5 \pmod{8}$.

    \item A family of binary $[2^s - 1, 2^{s-1}, d]$ cyclic codes for $s \equiv 3 \pmod{4}$ was derived in \cite{LIU}, where $d \geq 2^{(s-1)/2} + 1$ and $d^\perp \geq 2^{(s-1)/2} + 2$ if $s \equiv 3 \pmod{8}$, and $d \geq 2^{(s-1)/2} + 3$ and $d^\perp \geq 2^{(s-1)/2} + 4$ if $s \equiv 7 \pmod{8}$.

    \item A family of binary $[2^s - 1, 2^{s-1}, d]$ cyclic codes was given in \cite{SUN}, where $s \geq 9$ is odd, $d \geq 3 \cdot 2^{(s-1)/2} - 1$ and $d^\perp \geq 2^{\frac{s-1}{2}}$.

    \item A family of binary $[2^s - 1, 2^{s-1} \pm 1, d]$ cyclic codes was obtained in \cite{LWQ}, where $s$ is even, $d \geq 2^{s/2} - 1$ and $d^\perp \geq 2^{s/2}$.

    \item A family of binary $[2^s - 1, 2^{s-1}, d]$ cyclic codes was derived in \cite{LWQ}, where $s$ is odd, $d \geq 2^{(s+1)/2} - 1$ and $d^\perp \geq 2^{(s+1)/2}$.
    
    \item A family of binary $[2^s - 1, 2^{s-1}, d]$ cyclic codes was constructed in \cite{LWQ}, where $s$ is odd, $d \geq 2^{(s+3)/2} - 15$ and $d^\perp \geq 2^{(s-1)/2}$.
\end{itemize}
Among these known codes, only a few families of binary codes have simultaneously large minimal distances and large dual distances.
Therefore, the construction of new families of such codes is of great interest.

In this paper, for even $s$, we focus on new constructions of $[2^{s}-1,k]$ binary cyclic codes with rate approximately $1/2$,
for which both the minimum distance $d$ and the dual distance $d^\perp$ are simultaneously large.
  Some of the constructions yield binary cyclic codes with parameters $[2^s-1,2^{s-1}\pm 1,d]$ such that the lower bounds of $d\cdot d^\perp$ are close to $n$ or $2n$.  The other constructions yield binary cyclic codes with parameters $[2^s -1, 2^{s-1} + c, d]$ for some $ 1- \frac{5s}{2} \leq c \leq \frac{5s}{2} -1$ such
that the lower bounds of $d\cdot d^\perp$ are close to $n$ or $2n$. The parameters of the binary cyclic codes constructed by us and the known ones with simultaneously large minimal distances and large dual distances are summarized in Table \ref{tab}, where the notation $\approx k_0 n$ in the fifth column indicates that the lower bound of $d \cdot d^\perp$ equals $k_0 n + o(n)$ as $n \to \infty$ for $k_0 > 0$. Compared with the known binary cyclic codes with the same lengths and dimensions by Table \ref{tab}, most of our binary cyclic codes 
have larger lower bound of $d\cdot d^\perp$. Specifically, we have the following:
\begin{enumerate}
\item[$\bullet$] The lower bounds of $d\cdot d^\perp$ of our codes in Theorems \ref{th-16} and \ref{th-29} are about $2n$, while 
the lower bounds of $d\cdot d^\perp$ of the known codes with the same lengths and dimensions in \cite{CTCD} and \cite{LWQ} are much smaller than $2n$.
\item[$\bullet$] The lower bounds of $d\cdot d^\perp$ of our codes in Theorems \ref{th-1}, \ref{th-3}, \ref{th-22} and \ref{th-24} are about $2n$, while 
the lower bounds of $d\cdot d^\perp$ of the known codes with the same lengths and dimensions in \cite{SUN} are  only about  $3n/2$. 
\end{enumerate}

Particularly, we derive several codes with excellent parameters according to the Code Tables at http:\allowbreak//www.codetables.de/.

\begin{table*}[htbp]
\captionsetup{font=scriptsize}
\caption{Binary $[n,k,d]$ cyclic codes with minimum Hamming distance $d$, dual distance $d^\perp$ and $n=2^s-1$ for even $s$}
\label{tab}
\centering
\begin{tabular}{|c|c|c|c|c|c|}
\hline
$s$ & $k$ & $d \geq$ & $d^{\perp} \geq$ & The lower bound of $d \cdot d^{\perp}$& Reference \\
\hline
$s \equiv 4 \pmod{8}$ & $2^{s-1} + 1$ & $2^{(s-4)/2} + 1$ & $2^{(s-4)/2} + 2$ & $\approx \frac{1}{16} n$ & \cite{CTCD} \\
\hline
$s \equiv 2 \pmod{4}$ & $2^{s-1} + 1$ & $2^{(s-2)/2} + 1$ & $2^{(s-2)/2} + 2$ & $\approx \frac{1}{4} n$ & \cite{CTCD} \\
\hline
$s \equiv 0 \pmod{4}$ & $2^{s-1} - 1$ & $2^{(s-2)/2} + 1$ & $2^{(s-2)/2} + 2$ & $\approx \frac{1}{4} n$ & \cite{CTCD} \\
\hline
$s \equiv 2 \pmod{4}$ & $2^{s-1} - 1$ & $2^{(s-4)/2} + 1$ & $2^{(s-4)/2} + 2$ & $\approx \frac{1}{16} n$ & \cite{CTCD} \\
\hline
$s \equiv 2 \pmod{4}$ & $2^{s-1} + 1$ & $2^{(s+2)/2} - 1$ & $2^{s/2} $ & $\approx  2n$ & Theorem \ref{th-16} \\
\hline
$s \equiv 2 \pmod{4}$ & $2^{s-1} - 1$ & $2^{s/2} - 1$ & $2^{(s+2)/2} $ & $\approx  2n$ & Theorem \ref{th-16} \\
\hline
$s \equiv 2 \pmod{4}$ & $2^{s-1} + 1$ & $2^{(s+2)/2} - 3$ & $2^{s/2} $ & $\approx  2n$ & Theorem \ref{th-29} \\
\hline
$s \equiv 2 \pmod{4}$ & $2^{s-1} - 1$ & $2^{s/2} - 1$ & $2^{(s+2)/2}-2 $ & $\approx  2n$ & Theorem \ref{th-29} \\
\hline
$s \equiv 2 \pmod{4}$ & $2^{s-1} + 1$ & $2^{(s+2)/2} + 5$ & $2^{(s-2)/2} $ & $\approx  n$ & Theorem \ref{th-38} \\
\hline
$s \equiv 2 \pmod{4}$ & $2^{s-1} - 1$ & $2^{(s-2)/2} - 1$ & $2^{(s+2)/2}+ 6 $ & $\approx  n$ & Theorem \ref{th-38}\\
\hline
$s$ is even & $2^{s-1} - 1$ & $2^{\frac{s}{2}} - 1$ & $2^{\frac{s}{2}}$ & $\approx n$ & \cite{LWQ} \\
\hline
$s$ is even & $2^{s-1} + 1$ & $2^{\frac{s}{2}} - 1$ & $2^{\frac{s}{2}}$ & $\approx n$ & \cite{LWQ} \\
\hline
$s \equiv 2 \pmod{4}$  & $2^{s-1} +1 -\frac{s}{2}$ & $ 3 \cdot 2^{\frac{s-2}{2}} + 1$ & $2^{\frac{s}{2}}$ & $\approx \frac{3}{2} n$ & \cite{SUN} \\
\hline
$s \equiv 0 \pmod{4}$  & $2^{s-1} -1 +\frac{s}{2}$ & $ 3 \cdot 2^{\frac{s-2}{2}} + 1$ & $2^{\frac{s}{2}}$ & $\approx \frac{3}{2} n$ & \cite{SUN} \\
\hline
$s \equiv 2 \pmod{4}$  & $2^{s-1} -1 + \frac{s}{2}$ & $ 2^{\frac{s}{2}} -1$ & $3 \cdot 2^{\frac{s-2}{2}} + 2$ & $\approx \frac{3}{2} n$ & \cite{SUN} \\
\hline
$s \equiv 0 \pmod{4}$  & $2^{s-1} +1 - \frac{s}{2}$ & $ 2^{\frac{s}{2}} -1$ & $3 \cdot 2^{\frac{s-2}{2}} + 2$ & $\approx \frac{3}{2} n$ & \cite{SUN} \\
\hline
$s \equiv 2 \pmod{4}$  & $2^{s-1} + 1 - \frac{3s}{2}$ & $ 2^{\frac{s}{2}+1} -1$ & $2^{\frac{s}{2}}$ & $\approx  2n$ & Theorems \ref{th-1} and \ref{th-3} \\
\hline
$s \equiv 0 \pmod{4}$  & $2^{s-1} - 1 - \frac{s}{2}$ & $ 2^{\frac{s}{2}+1} -1$ & $2^{\frac{s}{2}}$ & $\approx  2n$ & Theorems \ref{th-1} and \ref{th-3} \\
\hline
$s \equiv 2 \pmod{4}$  & $2^{s-1} - 1 + \frac{3s}{2}$ & $ 2^{\frac{s}{2}} -1$ & $2^{\frac{s}{2}+1}$ & $\approx  2n$ & Theorems \ref{th-1} and \ref{th-3} \\
\hline
$s \equiv 0 \pmod{4}$  & $2^{s-1} + 1 + \frac{s}{2}$ & $ 2^{\frac{s}{2}} -1$ & $2^{\frac{s}{2}+1}$ & $\approx  2n$ & Theorems \ref{th-1} and \ref{th-3} \\
\hline
$s \equiv 2 \pmod{4}$  & $2^{s-1} + 1 - \frac{s}{2}$ & $ 2^{\frac{s}{2}+1} -3$ & $2^{\frac{s}{2}}$ & $\approx  2n$ & Theorems \ref{th-22} and \ref{th-24} \\
\hline
$s \equiv 0 \pmod{4}$  & $2^{s-1} - 1 + \frac{s}{2}$ & $ 2^{\frac{s}{2}+1} -3$ & $2^{\frac{s}{2}}$ & $\approx  2n$ & Theorems \ref{th-22} and \ref{th-24} \\
\hline
$s \equiv 2 \pmod{4}$  & $2^{s-1} - 1 + \frac{s}{2}$ & $ 2^{\frac{s}{2}} -1$ & $2^{\frac{s}{2}+1}-2$ & $\approx  2n$ & Theorems \ref{th-22} and \ref{th-24} \\
\hline
$s \equiv 0 \pmod{4}$  & $2^{s-1} + 1 - \frac{s}{2}$ & $ 2^{\frac{s}{2}} -1$ & $2^{\frac{s}{2}+1}-2$ & $\approx  2n$ & Theorems \ref{th-22} and\ref{th-24} \\
\hline
$s \equiv 2 \pmod{4}$  & $2^{s-1} + 1 - \frac{5s}{2}$ & $ 2^{\frac{s}{2}+1} +5$ & $2^{\frac{s}{2}-1}$ & $\approx  n$ & Theorems \ref{th-31} and \ref{th-33} \\
\hline
$s \equiv 0 \pmod{4}$  & $2^{s-1} - 1 - \frac{3s}{2}$ & $ 2^{\frac{s}{2}+1} +5$ & $2^{\frac{s}{2}-1}$ & $\approx  n$ & Theorems \ref{th-31} and \ref{th-33} \\
\hline
$s \equiv 2 \pmod{4}$  & $2^{s-1} - 1 + \frac{5s}{2}$ & $ 2^{\frac{s}{2}-1} -1$ & $2^{\frac{s}{2}+1}+6$ & $\approx  n$ & Theorems \ref{th-31} and \ref{th-33} \\
\hline
$s \equiv 0 \pmod{4}$  & $2^{s-1} + 1 + \frac{3s}{2}$ & $ 2^{\frac{s}{2}-1} -1$ & $2^{\frac{s}{2}+1}+6$ & $\approx  n$ & Theorems \ref{th-31} and \ref{th-33} \\
\hline
\end{tabular}
\end{table*}

\section{Preliminaries}
Throughout this paper, let $s$ be a positive integer and $n = 2^s - 1$. Let $\beta$ be a primitive element of $\mathbb{F}_{2^{s}}$. For a linear code $\mathcal{C}$, $\dim(\mathcal{C})$ and $d(\mathcal{C})$ denote its dimension and minimum Hamming distance, respectively.
\subsection{Cyclic codes and the BCH bound}
Let $\mathbb{Z}_n = \{0, 1, \ldots, n - 1\}$ be the ring of integers modulo $n$. For any $i \in \mathbb{Z}_n$, the $2$-cyclotomic coset of $i$ modulo $n$ is defined as
\[
C_i^{(2,n)} = \{ i \cdot 2^j \pmod{n} \mid 0 \leq j \leq \ell_i - 1 \},
\]
where $\ell_{i}$ is the smallest positive integer satisfying $i \cdot 2^{\ell_i} \equiv i \pmod{n}$. Hence $\ell_{i} = |C_i^{(2,n)}|$. The smallest integer of $C_i^{(2,n)}$ is referred to as its coset leader, and let $\Gamma_{(2,n)}$ denote the set of all such leaders. Then we have $$C_i^{(2,n)} \cap C_j^{(2,n)} = \emptyset$$ for any two distinct elements $i$ and $j$ in $\Gamma_{(2,n)}$, and
\[
\mathbb{Z}_{n} = \bigcup_{i \in \Gamma_{(2,n)}} C_{i}^{(2,n)}.
\]

Let $\mathcal{C}$ be a binary cyclic code of length $n$ with generator polynomial $g(x)$. Then there exists a subset $\Gamma \subseteq \Gamma_{(2,n)}$ such that
\[
g(x) = \prod_{i \in \Gamma} \mathbb{M}_{\beta^{i}}(x),
\]
which $\mathbb{M}_{\beta^{i}}(x)$ is the minimal polynomial of $\beta^{i}$ over $\mathbb{F}_{2}$.
The set $T := \bigcup_{i \in \Gamma} C_i^{(q,2n)}$ is called the defining set of $\mathcal{C}$ with respect to $\beta$. Then $\dim(\mathcal{C}) = n - |T|$. It is easy to deduce that the defining set of $\mathcal{C}^\perp$ is given by $-(\mathbb{Z}_{n} \setminus T) = \{-i : i \in \mathbb{Z}_{n} \setminus T\}$. The well known BCH bound is presented as follows. 

\begin{lemma}\cite{FS} [BCH Bound]
Let $\mathcal{C}$ be a binary cyclic code of length $n$ with defining set $T$. If there exist integers $b$ and $\delta$ with $2 \leq \delta \leq n$ such that
$\{i \bmod n : b \leq i \leq b + \delta - 2\} \subseteq T,$
then $d(\mathcal{C}) \geq \delta$.
\end{lemma}

\subsection{Extended codes}
Let $\mathcal{C}$ be an $[n,k,d]$ binary linear code. Its extended code $\overline{\mathcal{C}}$ is defined by
\[
\overline{\mathcal{C}} = \left\{(\textbf{c}, c_{n}): \textbf{c} = (c_0, c_1, \ldots, c_{n-1}) \in \mathcal{C}, \sum_{i=0}^{n} c_i=0\right\}.
\]
It is not hard to see that $\overline{\mathcal{C}}$ is an $[n+1,k,d_1]$ binary linear code, where $d_1 = d$ for even $d$, and $d_1 = d+1$ for odd $d$. 
The extending technique is useful for obtaining longer codes from old ones. In \cite{DT} and \cite{H}, this technique was employed to construct codes with interesting properties. 

\subsection{Some auxiliary results}
 For any $i$ with $1 \leq i \leq n-1$, define $\omega_2(i) = \sum_{h=0}^{s-1} i_h$, where $i = \sum_{h=0}^{s-1} i_h 2^h$ is the $2$-adic expansion of $i$ and each $i_h \in \Bbb Z_2$. Let $R_{(i,s)} := \{ h : 1 \leq h \leq n-1, \, \omega_2(h) \equiv i \pmod{2} \}$ for  $i \in \{0,1\}$.

\begin{lemma}\label{lem-3}\cite{LIU and Ding}
Let $s \geq 6$ be even, and let $i$ be an integer satisfying $1 \leq i \leq 2^{\frac{s+2}{2}} - 1$ and $i \equiv 1 \pmod2$. Then $i$ is a coset leader with $|C_i^{(2,n)}| = s$ except that $i = 2^{\frac{s}{2}} + 1$, in which case $|C_i^{(2,n)}| = \frac{s}{2}$.
\end{lemma}
\begin{lemma}\label{lem-4}\cite{CTCD}
Let $s \geq 6$ be even. Then $|R_{(0,s)}| = 2^{s-1} - 2$ and $|R_{(1,s)}| = 2^{s-1}$.
\end{lemma}

 \begin{lemma}\label{lem-6}
  Let $s \geq 6$ be even integer and let $i$ be an odd such that $1 \leq i \leq  2^{\frac{s}{2} + 1} - 1$. Then
\begin{eqnarray*}
\nonumber & & \left| \{ i : 1 \leq i \leq 2^{\frac{s}{2} + 1} - 1,\ i \equiv 1 \bmod{2},\  \omega_2(i) \equiv 0 \bmod{2}  \}\right|  \\
 &=&\left| \{ i : 1 \leq i \leq 2^{\frac{s}{2} + 1} - 1,\ i \equiv 1 \bmod{2},\  \omega_2(i) \equiv 1 \bmod{2}  \}\right| \\
 &=& 2^{\frac{s}{2} - 1 }.
\end{eqnarray*}
\end{lemma}

 \begin{proof}
 For $1 \leq i \leq 2^{\frac{s}{2} + 1} - 1$ and $i \equiv 1 \pmod{2}$, let $i = 1 + i_1 2 + i_2 2^2 + \cdots + i_{\frac{s}{2}} 2^{\frac{s}{2}} $, where $i_h \in \Bbb Z_2$. It is clear that $\omega_2(i) \equiv 0 \pmod{2}$ if and only if $\sum_{h=1}^{\frac{s}{2}} i_h \equiv 1 \pmod{2}$. Therefore, we deduce
\begin{eqnarray*}
\nonumber & & \left| \{ i : 1 \leq i \leq 2^{\frac{s}{2} + 1} - 1,\ i \equiv 1 \bmod{2},\ \omega_2(i) \equiv 0 \bmod{2} \} \right|\\
&=& \left| \left\{ (i_1, i_2, \ldots, i_{\frac{s}{2}}) \in \Bbb Z_2^{\frac{s}{2}} : \sum_{h=1}^{\frac{s}{2}} i_h \equiv 1 \pmod{2} \right\} \right|.
\end{eqnarray*}
 It is clear that the congruence $i_1 + i_2 + \ldots + i_{\frac{s}{2}}  \equiv 1 \pmod{2}$ is equivalent to 
 $$
 i_{\frac{s}{2}} \equiv 1+i_1 + i_2 + \ldots + i_{\frac{s}{2}-1}  \pmod{2}.
 $$
 Then we have
 $$
 \left| \left\{ (i_1, i_2, \ldots, i_{\frac{s}{2}}) \in \Bbb Z_2^{\frac{s}{2}} : \sum_{h=1}^{\frac{s}{2}} i_h \equiv 1 \pmod{2} \right\} \right|= 2^{\frac{s}{2} -1}.
 $$
Meanwhile, we have
 \begin{eqnarray*}
 \nonumber& & \left| \{ i : 1 \leq i \leq 2^{\frac{s}{2} + 1} - 1,\ i \equiv 1 \pmod{2}  \}\right|\\
 & = &  \left| \{ i : 1 \leq i \leq 2^{\frac{s}{2} + 1} - 1,\ i \equiv 1 \pmod{2},\ \omega_2(i) \equiv 0 \pmod{2}  \}\right|\\
 &&+\left| \{ i : 1 \leq i \leq 2^{\frac{s}{2} + 1} - 1,\  i \equiv 1 \pmod{2},\ \omega_2(i) \equiv 1 \pmod{2}  \}\right| \\
 & = & 2^{\frac{s}{2}}.
 \end{eqnarray*}
Thus,
 \begin{eqnarray*}
& & \left \{ i : 1 \leq i \leq 2^{\frac{s}{2} + 1} - 2,\  i \equiv 1 \pmod{2},\ \omega_2(i) \equiv 1 \pmod{2}  \} \right|\\
&=& 2^{\frac{s}{2} - 1 }.
 \end{eqnarray*}
 This completes the proof.
 \end{proof}

\section{New constructions of binary cyclic codes}
 In this section, for even $s$, we present new constructions of $[2^{s}-1,k]$ binary cyclic codes whose rates are bout $1/2$ and minimal distances $d$ and dual distances $d^\perp$ are simultaneously large. 
 
 \subsection{The first construction of binary cyclic codes with  $ d \cdot d^\perp \approx 2n$}
 Let $n = 2^s - 1$ and $s \geq 6$ be even. we have
\[
R_{(i,s)} = \{h : 1 \leq h \leq n-1, \ \omega_2(h) \equiv i \pmod{2}\}
\]
for $i \in \{0,1\}$. Denote
\[
A = \bigcup_{h=1}^{2^{\frac{s}{2}+1}-2} C_{h}^{(2,n)}
\]
and
\[
B_{(i,s)} = \bigcup_{\substack{h=1 \\ \omega_2(h) \equiv i \pmod{2}}}^{2^{\frac{s}{2}+1}-2} C_{n-h}^{(2,n)}
\]
for each $i \in \{0,1\}$. Let
\[
D_{(1,s)} = 
\begin{cases}
(R_{(0,s)} \setminus B_{(0,s)}) \cup A & \text{if } s \equiv 2 \pmod{4}, \\
(R_{(1,s)} \setminus B_{(1,s)}) \cup A & \text{if } s \equiv 0 \pmod{4},
\end{cases}
\]
and $D_{(0,s)} = \{1,2,\cdots,n-1\} \setminus D_{(1,s)}$.
 Then both $D_{(0,s)}$ and $D_{(1,s)}$ consist of unions of certain $2$-cyclotomic cosets.
For  $i \in \{0,1\}$, let $\mathcal{C}{(i,s)}$ be the binary cyclic code of length $n$ with defining set  $D_{(i,s)}$. The following lemmas are used to characterize the parameters of $\mathcal{C}_{(i,s)}$.

\begin{lemma}\label{lem-5}
Let $s \geq 6$ be even and $1 \leq i \leq  2^{\frac{s}{2} + 1} - 2$ be an integer. Then
\begin{eqnarray*}
 \left(\bigcup_{i=1}^{2^{\frac{s}{2} + 1} - 2} C_i^{(2,n)}\right) \bigcap  \left(\bigcup_{h=1}^{2^{\frac{s}{2} + 1} - 2} C_{n-h}^{(2,n)}\right) \\
= C_{2^{\frac{s}{2}}-1}^{(2,n)} \cup C_{3 \cdot 2^{\frac{s}{2}-1}-1}^{(2,n)} \cup C_{2^{\frac{s}{2}+1}-3}^{(2,n)}.
\end{eqnarray*}
\end{lemma}

\begin{proof}
For $1 \leq i \leq  2^{\frac{s}{2} + 1} - 2$, we have $C_i^{(2,n)} \subseteq \bigcup_{h=1}^{2^{\frac{s}{2} + 1} - 2} C_{n-h}^{(2,n)}$ if and only if $i \in C_{n-h}^{(2,n)}$ for some $1 \leq h \leq  2^{\frac{s}{2} + 1} - 2$. Equivalently, one can find an integer $k$ with $0 \leq k \leq s - 1$ such that
\begin{eqnarray}\label{eqn-1}
i + h \cdot 2^k \equiv 0 \pmod{n}.
\end{eqnarray}
 
Case 1: If $0 \leq k \leq \frac{s}{2}-2$, then for any $1 \leq i, h \leq 2^{\frac{s}{2} + 1} - 2$, it holds that
\begin{eqnarray*}
1 < i + h \cdot 2^k &\leq (2^{\frac{s}{2}+1}-2)(2^{\frac{s}{2}-2}+1)\\
&=2^{s-1} + 3\cdot 2^\frac{s-2}{2} -2 < n.
\end{eqnarray*}
This implies that Equation \eqref{eqn-1} has no solution.

Case 2: If $k = \frac{s}{2}-1$, then for any $1 \leq i, h \leq 2^{\frac{s}{2} + 1} - 2$, we have
\begin{eqnarray*}
1 < i + h \cdot 2^k &\leq (2^{\frac{s}{2}+1}-2)(2^{\frac{s}{2}-1}+1)\\
&= 2^s + 2^\frac{s}{2} -2 <2n.
\end{eqnarray*}
In this case, Equation \eqref{eqn-1} is equivalent to $i + h \cdot 2^{\frac{s}{2}-1} = 2^s - 1$. Therefore, the equation has the following solutions: $(i, h) = (2^{\frac{s}{2}} - 1, 2^{\frac{s}{2}+1} - 2)$ and $(3 \cdot 2^{\frac{s}{2}-1} - 1, 2^{\frac{s}{2}+1} - 3)$.

Case 3: If $k = \frac{s}{2}$, then for any $1 \leq i, h \leq 2^{\frac{s}{2} + 1} - 2$, we have
\begin{eqnarray*}
1 < i + h \cdot 2^k &\leq (2^{\frac{s}{2}+1}-2)(2^\frac{s}{2}+1)\\
&= 2^{s+1} -2  = 2n.
\end{eqnarray*}
In this case, Equation \eqref{eqn-1} holds if and only if $i + h \cdot 2^{\frac{s}{2}} = 2^s - 1$ or $i + h \cdot 2^{\frac{s}{2}} = 2^{s+1} - 2$, where the first equation has the unique solution $(i, h) = (2^{\frac{s}{2}} - 1, 2^{\frac{s}{2}} - 1)$ and the second one has no solution.

Case 4: If $k = \frac{s}{2}+1$, then $s - k = \frac{s}{2} - 1$. Similarly to Case 2, we deduce Equation \eqref{eqn-1} has solutions $(i, h) = ( 2^{\frac{s}{2}+1} - 2, 2^{\frac{s}{2}} - 1)$ and $(i,h)=(2^{\frac{s}{2}+1} - 3,3 \cdot 2^{\frac{s}{2}-1} - 1)$.

Case 5: If $ \frac{s}{2} + 2 \leq k \leq s-1$, then $1 \leq s-k \leq \frac{s}{2}-2$. Similarly to Case 1, we deduce Equation \eqref{eqn-1} has no solution.

Therefore, the desired conclusion is obtained from the above discussion. 
\end{proof}
 
 \begin{lemma}\label{lem-7}
   Let $s \geq 6$ be even. Then the following conclusions are obtained.
   \begin{enumerate}
    \item  $B_{(i,s)} \subseteq R_{(i,s)}$ for  $i \in \{0, 1\}$.
    
    \item If $s \equiv 2 \pmod{4}$, then $B_{(0,s)} \cap A = \emptyset$ and $B_{(1,s)} \cap A = C_{2^{\frac{s}{2}}-1}^{(2,n)} \cup C_{3 \cdot 2^{\frac{s}{2}-1}-1}^{(2,n)} \cup C_{2^{\frac{s}{2}+1}-3}^{(2,n)}$.
    
    \item If $s \equiv 0 \pmod{4}$, then $B_{(0,s)} \cap A = C_{2^{\frac{s}{2}}-1}^{(2,n)} \cup C_{3 \cdot 2^{\frac{s}{2}-1}-1}^{(2,n)} \cup C_{2^{\frac{s}{2}+1}-3}^{(2,n)}$ and $B_{(1,s)} \cap A = \emptyset$.
\end{enumerate}
 \end{lemma}
 \begin{proof}
1) Observe that $\omega_2(n - h) = s - \omega_2(h)$ for all $0 \leq h \leq n - 1$. Since $s$ is even, we obtain
\[
\omega_2(n - h) \equiv \omega_2(h) \pmod{2}.
\]
Therefore, $B_{(i,s)} \subseteq R_{(i,s)}$. 

2) Note that
\[
B_{(0,s)} \cup B_{(1,s)} = \bigcup_{j=1}^{2^{\frac{s}{2} + 1} - 2} C_{n-h}^{(2,n)}.
\]
It follows from Lemma 5 that
\begin{eqnarray*}
\nonumber & &  A \cap (B_{(0,s)} \cup B_{(1,s)})\\
&=& (A \cap B_{(0,s)}) \cup (A \cap B_{(1,s)})\\
&=&  C_{2^{\frac{s}{2}}-1}^{(2,n)} \cup C_{3 \cdot 2^{\frac{s}{2}-1}-1}^{(2,n)} \cup C_{2^{\frac{s}{2}+1}-3}^{(2,n)}.
\end{eqnarray*}
Observe that $B_{(0,s)} \cap A$ and $B_{(1,s)} \cap A$ consist of unions of certain $2$-cyclotomic cosets modulo $n$. Moreover,  $\omega_2(2^{\frac{s}{2}}-1) = \omega_2(3 \cdot 2^{\frac{s}{2}-1}-1) = \omega_2(2^{\frac{s}{2}+1}-3) = \frac{s}{2}$.

If $s \equiv 2 \pmod{4}$, then $\omega_2(2^{\frac{s}{2}}-1) = \omega_2(3 \cdot 2^{\frac{s}{2}-1}-1) = \omega_2(2^{\frac{s}{2}+1}-3) = \frac{s}{2} \equiv 1 \pmod{2}$ . It then follows that $B_{(0,s)} \cap A = \emptyset$ and $B_{(1,s)} \cap A = C_{2^{\frac{s}{2}}-1}^{(2,n)} \cup C_{3 \cdot 2^{\frac{s}{2}-1}-1}^{(2,n)} \cup C_{2^{\frac{s}{2}+1}-3}^{(2,n)}$.

3) If $s \equiv 0 \pmod{4}$, then $\omega_2(2^{\frac{s}{2}}-1) = \omega_2(3 \cdot 2^{\frac{s}{2}-1}-1) = \omega_2(2^{\frac{s}{2}+1}-3) = \frac{s}{2} \equiv 0 \pmod{2}$. This implies that $B_{(0,s)} \cap A = C_{2^{\frac{s}{2}}-1}^{(2,n)} \cup C_{3 \cdot 2^{\frac{s}{2}-1}-1}^{(2,n)} \cup C_{2^{\frac{s}{2}+1}-3}^{(2,n)}$ and $B_{(1,s)} \cap A = \emptyset$.
\end{proof}

 \begin{lemma}\label{lem-8}
Let $s \geq 6$ be even. Then we have the following conclusions. 
\begin{enumerate}
    \item $\{h : 1 \leq h \leq  2^{\frac{s}{2}+1} -2\} \subseteq D_{(1,s)}$.
    
    \item $\{n - h : 1 \leq h \leq 2^{\frac{s}{2}} - 2\} \subseteq D_{(0,s)}$.
    
    \item If $s \equiv 2 \pmod{4}$, then $|D_{(1,s)}| = 2^{s-1} - 2 + \frac{3s}{2}$ and $|D_{(0,s)}| = 2^{s-1} - \frac{3s}{2}$.
    
    \item If $s \equiv 0 \pmod{4}$, then $|D_{(1,s)}| = 2^{s-1} + \frac{s}{2}$ and $|D_{(0,s)}| = 2^{s-1} - 2 - \frac{s}{2}$.
\end{enumerate}
\end{lemma}

 \begin{proof}
 1)   Observe that $\{ h : 1 \leq h \leq 2^{\frac{s}{2}+1} - 2\} \subseteq A \subseteq D_{(1,s)}$.
 
 2) Let $1 \leq h \leq 2^{\frac{s}{2}} - 2$, Note that $n - h \in D_{(0,s)}$ if and only if $n - h \notin D_{(1,s)}$.  It is obvious that  $n - h \notin A$ for each $1 \leq h \leq 2^{\frac{s}{2}} - 2$ by Lemma \ref{lem-5}. Hence, it suffices to prove that $n - h \notin (R_{(i,s)} \setminus B_{(i,s)})$ for each $1 \leq h \leq 2^{\frac{s}{2}} - 2$ and $i \in \{0,1\}$. 

Case 1: If $\omega_2(h) \equiv i \pmod{2}$, then $n - h \in B_{(i,s)}$. Therefore, $n - h \notin (R_{(i,s)} \setminus B_{(i,s)})$.

Case 2: If $\omega_2(h) \equiv i + 1 \pmod{2}$, then $n - h \in R_{(i \oplus 1,s)}$, where $\oplus$ denotes the modulo 2 addition. Note that $R_{(i,s)} \cap R_{(i \oplus 1,s)} = \emptyset$, then $n - h \notin (R_{(i,s)} \setminus B_{(i,s)})$.

Therefore, the desired conclusion is obtained from the above discussion. 

3) Let $s \equiv 2 \pmod{4}$. By Lemma \ref{lem-7}, we have $B_{(0,s)} \subseteq R_{(0,s)}$ and $B_{(0,s)} \cap A = \emptyset$. Thus,
\begin{eqnarray}\label{eqn-S}
\nonumber & &(R_{(0,s)} \setminus B_{(0,s)}) \cap A \\
 \nonumber        &=& (R_{(0,s)} \cap A ) \setminus (B_{(0,s)} \cap A ) \\
 \nonumber        &=& R_{(0,s)} \cap A\\
          &=& \bigcup_{\substack{h=1 \\ \omega_2(h) \equiv 0 \pmod{2}}}^{2^{\frac{s}{2}+1} - 2} C_h^{(2,n)}.
\end{eqnarray}
By Lemma \ref{lem-3}, every integer $h$ satisfying $1 \leq h \leq 2^{\frac{s}{2}+1} - 2$ and $h \equiv 1 \pmod{2}$ is a coset leader and $|C_h^{(2,n)}| = s$ if $h \neq 2^{\frac{s}{2}} + 1$. For each $h$ that satisfies $1 \leq h \leq 2^{\frac{s}{2}+1} - 2$ and $h \not\equiv 1 \pmod{2}$,  there exists $h'<h$ such that  $h' \equiv 1 \pmod{2}$ and $ h \in C_{h'}^{(2,n)}$. Observe that $\omega_2(2^{\frac{s}{2}} + 1) =  2$, $|C_{2^{\frac{s}{2}} + 1}^{(2,n)}| = \frac{s}{2}$ and $\omega_2(2^{\frac{s}{2}+1} - 1) \equiv 0 \pmod{2}$. By Equation (\ref{eqn-S}) and Lemma \ref{lem-6}, we then have
\begin{eqnarray*}
\nonumber & & |(R_{(0,s)} \setminus B_{(0,s)}) \cap A| \\
          &=&  s \cdot \left|\left\{1 \leq h \leq 2^{\frac{s}{2}+1} - 2: h \equiv 1 \pmod{2},   \omega_2(h) \equiv 0 \pmod{2}\right\}\right| - \frac{s}{2}\\
          &=&  s \cdot \left|\left\{1 \leq h \leq 2^{\frac{s}{2}+1} - 1: h \equiv 1 \pmod{2}, \omega_2(h) \equiv 0 \pmod{2}\right\}\right|-s- \frac{s}{2}\\
          &=& s \cdot 2^{\frac{s}{2}-1} - \frac{3s}{2}
\end{eqnarray*}
and
\begin{eqnarray}\label{A}
|A| = s \cdot (2^{\frac{s}{2}} - 1) - \frac{s}{2} = s \cdot 2^{\frac{s}{2}} - \frac{3s}{2}.
\end{eqnarray}
 Note that $C_{n-h_1}^{(2,n)} = C_{n-h_2}^{(2,n)}$ if and only if $C_{h_1}^{(2,n)} = C_{h_2}^{(2,n)}$ and $|C_{n-h}^{(2,n)}| = |C_h^{(2,n)}|$. Then
\begin{eqnarray*}
 \nonumber |B_{(0,s)}| &=& \left| \bigcup_{\substack{h=1 \\ \omega_2(h) \equiv 0 \pmod{2}}}^{2^{\frac{s}{2}+1} - 2} C_h^{(2,n)} \right|\\
                       &=& |(R_{(0,s)} \setminus B_{(0,s)}) \cap A|.
\end{eqnarray*}
Therefore, we have
\begin{eqnarray*}
 \nonumber |D_{(1,s)}| &=& |(R_{(0,s)} \setminus B_{(0,s)}) \cup A| \\
   &=& |R_{(0,s)} \setminus B_{(0,s)}| + |A| - |(R_{(0,s)} \setminus B_{(0,s)}) \cap A|\\
   &=& |R_{(0,s)}| + |A| - |B_{(0,s)}| - |(R_{(0,s)} \setminus B_{(0,s)}) \cap A|\\
   &=& |R_{(0,s)}| + \frac{3s}{2} \\
   &=& 2^{s-1} - 2 + \frac{3s}{2}
\end{eqnarray*}
by Lemma \ref{lem-4}.
In addition, we have $|D_{(0,s)}| = n - 1 - |D_{(1,s)}| = 2^{s-1} - \frac{3s}{2}$.

4) Let $s \equiv 0 \pmod{4}$. $B_{(1,s)} \subseteq R_{(1,s)}$ and $B_{(1,s)} \cap A = \emptyset$  by Lemma \ref{lem-7}. Then we have
\begin{eqnarray*}
\nonumber & &(R_{(1,s)} \setminus B_{(1,s)}) \cap A \\
          &=& (R_{(1,s)} \cap A ) \setminus (B_{(1,s)} \cap A ) \\
          &=& R_{(1,s)} \cap A\\
          &=& \bigcup_{\substack{h=1 \\ \omega_2(h) \equiv 1 \pmod{2}}}^{2^{\frac{s}{2}+1} - 2} C_h^{(2,n)}.
\end{eqnarray*}
Observe that $\omega_2(2^{\frac{s}{2}+1} - 1) \equiv 1 \pmod{2}$ and  $\omega_2(2^{\frac{s}{2}} + 1) \equiv 0 \pmod{2}$. By Lemmas \ref{lem-3} and \ref{lem-6}, we then have
\begin{eqnarray*}
\nonumber & & |R_{(1,s)} \cap A|=|B_{(1,s)}| \\
&=&\left| \bigcup_{\substack{h=1 \\ \omega_2(h) \equiv 1 \pmod{2}}}^{2^{\frac{s}{2}+1} - 2} C_h^{(2,n)} \right|\\
          &=&  s \cdot |\{1 \leq h\leq 2^{\frac{s}{2}+1} - 2: h  \equiv 1 \pmod{2}, \omega_2(h) \equiv 1 \pmod{2}\}| \\
          &=&  s \cdot |\{1 \leq h \leq 2^{\frac{s}{2}+1} - 1: h \equiv 1 \pmod{2},\omega_2(h) \equiv 1 \pmod{2}\}|-s\\
          &=& s \cdot 2^{\frac{s}{2}-1} - s.
\end{eqnarray*}
Therefore, we have
\begin{eqnarray*}
 \nonumber |D_{(1,s)}| &=& |(R_{(1,s)} \setminus B_{(1,s)}) \cup A| \\
   &=& |R_{(1,s)} \setminus B_{(1,s)}| + |A| - |(R_{(1,s)} \setminus B_{(1,s)}) \cap A|\\
   &=& |R_{(1,s)}| + |A| - |B_{(1,s)}| - |R_{(1,s)} \cap A|\\
   &=& 2^{s-1} + \frac{s}{2}
\end{eqnarray*}
by Lemma \ref{lem-4} and Equation (\ref{A}).
In addition, we have $|D_{(0,s)}| = n - 1 - |D_{(1,s)}| = 2^{s-1} - 2 - \frac{s}{2}$.
 \end{proof}
 
The parameters of $\mathcal{C}_{(i,s)}$ are characterized as follows.
 
\begin{theorem}\label{th-1}
Let $s \geq 6$ be even. Then the following results are obtained.
\begin{enumerate}
    \item If $s \equiv 2 \pmod{4}$, then the code $\mathcal{C}_{(1,s)}$ has parameters
    $
    \left[ 2^s - 1, 2^{s-1} + 1 - \frac{3s}{2}, d \geq 2^{\frac{s}{2} + 1 } - 1 \right],
   $
    and  the code $\mathcal{C}_{(0,s)}$ has parameters
    $
    \left[ 2^s - 1, 2^{s-1} - 1 + \frac{3s}{2}, d \geq 2^{\frac{s}{2}} - 1 \right].
    $
    
    \item If $s \equiv 0 \pmod{4}$, then the code $\mathcal{C}_{(1,s)}$ has parameters
    $
    \left[ 2^s - 1, 2^{s-1} - 1 -\frac{s}{2}, d \geq  2^{\frac{s}{2} + 1} - 1 \right],
   $
    and the code $\mathcal{C}_{(0,s)}$ has parameters
   $
    \left[ 2^s - 1, 2^{s-1} + 1 + \frac{s}{2}, d \geq 2^{\frac{s}{2}} - 1 \right].
    $
\end{enumerate}
\end{theorem}

\begin{proof}
By 1) and 2) of Lemma \ref{lem-8}, the defining set $D_{(1,s)}$ of the code $\mathcal{C}_{(1,s)}$ contains the set $\{h : 1 \leq h \leq  2^{\frac{s}{2}+1} -2\}$, while the defining set $D_{(0,s)}$ of  $\mathcal{C}_{(0,s)}$  contains the set $\{n - h : 1 \leq h \leq 2^{\frac{s}{2}} - 2\}$. The BCH bound implies that
$
d(\mathcal{C}_{(1,s)}) \geq 2^{\frac{s}{2} + 1 } - 1 \quad \text{and} \quad d(\mathcal{C}_{(0,s)}) \geq 2^{\frac{s}{2}} - 1,
$
respectively. Note that $\dim(\mathcal{C}_{(i,s)}) = n - |D_{(i,s)}|$ for $i \in \{0,1\}$, then the dimensions of $\mathcal{C}_{(1,s)}$ and $\mathcal{C}_{(0,s)}$  follow from Lemma \ref{lem-8}. 
\end{proof}

\begin{theorem}\label{th-2}
Let $s \geq 6$ be even. Then the following conclusions are obtained.
\begin{enumerate}
    \item If $s \equiv 2 \pmod{4}$, then  the code $\overline{\mathcal{C}_{(1,s)}}$ has parameters
    $
    \left[ 2^s, 2^{s-1} + 1 - \frac{3s}{2}, d \geq 2^{\frac{s}{2} + 1 }  \right],
    $
    and the code $\overline{\mathcal{C}_{(0,s)}}$ has parameters
    $
    \left[ 2^s - 1, 2^{s-1} - 1 - \frac{3s}{2}, d \geq 2^{\frac{s}{2}} \right].
    $
    
    \item If $s \equiv 0 \pmod{4}$, then  the code $\overline{\mathcal{C}_{(1,s)}}$ has parameters
    $
    \left[ 2^s, 2^{s-1} - 1 -\frac{s}{2}, d \geq  2^{\frac{s}{2} + 1}  \right],
    $
    and  the code $\overline{\mathcal{C}_{(0,s)}}$ has parameters
    $
    \left[ 2^s, 2^{s-1} + 1 + \frac{s}{2}, d \geq 2^{\frac{s}{2}}  \right].
    $
\end{enumerate}
\end{theorem}
\begin{proof}
  The desired results can be deduced from Theorem \ref{th-1} and the definition of the extended codes.
\end{proof}

\begin{theorem}\label{th-3}
Let $s \geq 6$ be even. Then the following conclusions are obtained.
\begin{enumerate}
    \item If $s \equiv 2 \pmod{4}$, then the dual code $\mathcal{C}^{\perp}_{(1,s)}$ has
parameters
    $
    \left[ 2^s - 1, 2^{s-1} +\frac{3s}{2} - 2, d^\perp \geq 2^{\frac{s}{2} }  \right],
    $
    and the dual code $\mathcal{C}^{\perp}_{(0,s)}$  has parameters

    $
    \left[ 2^s - 1, 2^{s-1} - \frac{3s}{2}, d^\perp \geq 2^{\frac{s}{2}+1}  \right].
    $
    
    \item If $s \equiv 0 \pmod{4}$, then the dual code $\mathcal{C}^{\perp}_{(1,s)}$ has
parameters
    $
    \left[ 2^s - 1, 2^{s-1} +  \frac{s}{2}, d^\perp \geq  2^{\frac{s}{2}}  \right],
    $
    and the dual code $\mathcal{C}^{\perp}_{(0,s)}$ has
parameters
    $
    \left[ 2^s - 1, 2^{s-1} -  \frac{s}{2} -2, d^\perp \geq 2^{\frac{s}{2}+1}  \right].
    $
\end{enumerate}
\end{theorem}
\begin{proof}
Since $\dim(\mathcal{C}_{(i,s)}) + \dim(\mathcal{C}^{\perp}_{(i,s)}) = n $, the  dimension of $\dim(\mathcal{C}_{(i,s)})$ for each $i\in \{0,1\}$ follows from Theorem \ref{th-1}. 
The defining set of $ \mathcal{C}^{\perp}_{(i,s)}$ equals  $\mathbb{Z}_{n} \setminus (-D_{(i,s)})$. The remaining proofs proceed by considering the following cases.

Case 1: If $i =1$.  We deduce the set $\{h : 0 \leq h \leq 2^{\frac{s}{2}} - 2\} \subseteq \mathbb{Z}_{n} \setminus (-D_{(1,s)})$ by 2) of Lemma \ref{lem-8}.
 By the BCH bound,  the desired lower bound on  $d(\mathcal{C}^{\perp}_{(1,m)})$ is obtained.
 
Case 2: If $i =0$.  We deduce the set $\{n-h : 0 \leq h \leq 2^{\frac{s}{2}+1} - 2\} \subseteq \mathbb{Z}_{n} \setminus (-D_{(0,s)})$ by 1) of Lemma \ref{lem-8}.
 By the BCH bound,  the desired lower bound on $d(\mathcal{C}^{\perp}_{(0,s)})$ is obtained.
\end{proof}

\begin{example}
We remark that the bounds on the minimum distances of the codes in this subsection are tight in some cases. Now we give some examples by Magma. 
  When $s = 6$, then $n = 2^6 - 1 = 63$. 
Then the binary cyclic code $\mathcal{C}_{(1,s)}$ has parameters $\left[63, 24, 15 \right]$, where the minimum distance is very close to that of the best known $[63,24,16]$ binary linear code according to the Database\cite{Grassl}. The extended code $\overline{\mathcal{C}_{(1,s)}}$ has parameters $\left[64, 24, 16 \right]$ which are the same as the best-known parameters according to the Database\cite{Grassl}. The binary cyclic code  $\mathcal{C}_{(0,s)}$ has parameters $\left[63, 40, 7 \right]$,   where the minimum distance is very close to that of the best known $[63, 40, 8]$ binary linear code according to the Database\cite{Grassl}. 
The extended code $\overline{\mathcal{C}_{(0,s)}}$ has parameters $\left[64, 40, 8 \right]$, where the minimum distance is very close to that of the best known $[64,40,9]$ binary linear code according to the Database\cite{Grassl}. The binary cyclic code $\mathcal{C}^{\perp}_{(1,s)}$ has parameters $\left[63, 39, 8 \right]$, where the minimum distance is very close to that of the best known $[63,39,9]$ binary linear code according to the Database\cite{Grassl}.  The code $\mathcal{C}^{\perp}_{(0,s)}$ has parameters $\left[63, 23, 16 \right]$ which are the same as the best-known parameters according to the Database\cite{Grassl}.
\end{example}

 \subsection{The second construction of binary cyclic codes with  $ d \cdot d^\perp \approx 2n$}
Let the notations follow the above. In the following, we modify the first construction. 
Let $s \equiv 2 \pmod{4}$, 
\[
D_{(1,s)}^\prime := \left[R_{(0,s)} \setminus \left(B_{(0,s)} \cup C_{5 \cdot 2^\frac{s}{2} +5}^{(2,n)} \cup  C_{3 \cdot 2^{\frac{s}{2}}-13}^{(2,n)}\right)\right] \cup A 
\]
and $D_{(0,s)}^\prime: = \{1,2,\cdots,n-1\} \setminus D_{(1,s)}^\prime$. Then both $D_{(0,s)}^\prime$ and $D_{(1,s)}^\prime$ consist of unions of certain $2$-cyclotomic cosets. 
Let $\mathcal{C}_{(i,s)}^\prime$ be the binary cyclic code of length $n$ with defining set $D_{(i,s)}^\prime$ for $i \in \{0,1\}$. The following lemmas are used to characterize the parameters of $\mathcal{C}_{(i,s)}^\prime$.
\begin{lemma}\label{lem-14}
  Let $s \equiv 2 \pmod{4}$ and $s \geq 10$. The following conclusions are obtained.
  \begin{enumerate}
    \item $C_{5 \cdot 2^\frac{s}{2} +5}^{(2,n)} \cup C_{3 \cdot 2^{\frac{s}{2}}-13}^{(2,n)} \subseteq R_{(0,s)}$.
    \item The coset leader of $C_{5 \cdot 2^\frac{s}{2} +5}^{(2,n)}$ is $5 \cdot 2^\frac{s}{2} +5$ and $\left|C_{5 \cdot 2^\frac{s}{2} +5}^{(2,n)}\right| = \frac{s}{2}$. The coset leader of $C_{3 \cdot 2^{\frac{s}{2}}-13}^{(2,n)}$ is $3 \cdot 2^{\frac{s}{2}}-13$ and $\left|C_{3 \cdot 2^{\frac{s}{2}}-13}^{(2,n)}\right| = s$.
    \item $\left(C_{5 \cdot 2^\frac{s}{2} +5}^{(2,n)} \cup C_{3 \cdot 2^{\frac{s}{2}}-13}^{(2,n)}\right) \cap A=\emptyset$.
    \item $\left(C_{5 \cdot 2^\frac{s}{2} +5}^{(2,n)} \cup C_{3 \cdot 2^{\frac{s}{2}}-13}^{(2,n)}\right) \cap  B_{(0,s)}=\emptyset$.
  \end{enumerate}
\end{lemma}
\begin{proof}
  1) It is easy to verify that $3 \cdot 2^{\frac{s}{2}} - 13 = 2^{\frac{s}{2}+1} + \sum_{h=4}^{\frac{s}{2}-1} 2^{h} + 2+1$. Then $\omega_2(3 \cdot 2^{\frac{s}{2}} - 13)= \frac{s-2}{2}$.  Since $s \equiv 2 \pmod{4}$, $\omega_2(3 \cdot 2^{\frac{s}{2}} - 13) \equiv 0 \pmod{2}$. 
  It is clear that $\omega_2(5 \cdot 2^\frac{s}{2} +5) = 4 \equiv 0 \pmod{2}$. Then $C_{5 \cdot 2^\frac{s}{2} +5}^{(2,n)}  \cup C_{3 \cdot 2^{\frac{s}{2}}-13}^{(2,n)} \subseteq S_{(0,s)}$.
  
  2) By the  2-adic expansion of $5 \cdot 2^\frac{s}{2} +5$ and $3 \cdot 2^{\frac{s}{2}}-13$, it can be verified that they are coset leaders, with  $\left|C_{5 \cdot 2^\frac{s}{2} +5}^{(2,n)}\right| = \frac{s}{2}$ and $\left|C_{3 \cdot 2^{\frac{s}{2}}-13}\right| = s$.
  
  3) By Lemma 3, every integer $i$ with $1 \leq i \leq 2^{\frac{s}{2}+1} - 1$ and $i \equiv 1 \pmod{2}$ is a coset leader. Since $5 \cdot 2^\frac{s}{2} +5$ is a coset leader and $5 \cdot 2^\frac{s}{2} +5 > 2^{\frac{s}{2}+1} -1$, we deduce $C_{5 \cdot 2^\frac{s}{2} +5}^{(2,n)} \cap A=\emptyset$. Similarly, we deduce $C_{3 \cdot 2^{\frac{s}{2}}-13}^{(2,n)} \cap A =\emptyset$.
  
  4) If $C_{5 \cdot 2^\frac{s}{2} +5}^{(2,n)} \cap B_{(0,s)}\neq \emptyset$, then there exists a $1 \leq h \leq 2^{\frac{s}{2}+1} - 4$ and $0 \leq l \leq s - 1$ such that
\[
(5 \cdot 2^\frac{s}{2} +5) \cdot 2^l \bmod n = n - h.
\]
Note that $(5 \cdot 2^\frac{s}{2} +5) \cdot 2^l \bmod n \leq 2^{s-1} + 2^{s-3} + 2^{\frac{s}{2}-1} + 2^{\frac{s}{2}-3}$ for any $0 \leq l \leq s - 1$. When $1 \leq h \leq 2^{\frac{s}{2}+1} - 4$, then $n - h \geq 2^s - 2^{\frac{s}{2}+1} + 3$.
It follows that
\[
(5 \cdot 2^\frac{s}{2} +5) \cdot 2^l \bmod n < n - h
\]
for all $1 \leq h \leq 2^{\frac{s}{2}+1} - 4$ and $0 \leq l \leq s - 1$, which leads to a contradiction. Therefore, $C_{5 \cdot 2^\frac{s}{2} +5}^{(2,n)} \cap B_{(0,s)}= \emptyset$.
Similarly, we have $C_{3 \cdot 2^{\frac{s}{2}}-13}^{(2,n)} \cap B_{(0,s)}=\emptyset$.
\end{proof}

\begin{lemma}\label{lem-15}
Let $s \equiv 2 \pmod 4$ and $s \geq 10$. The following results hold.
\begin{enumerate}
    \item $\{h : 1 \leq h \leq  2^{\frac{s}{2}+1} -2\} \subseteq D_{(1,s)}^\prime$.
    
    \item $\{n - h : 1 \leq h \leq 2^{\frac{s}{2}} - 2\} \subseteq D_{(0,s)}^\prime$.
    
    \item $|D_{(1,s)}^\prime| = 2^{s-1} - 2 $ and $|D_{(0,s)}^\prime| = 2^{s-1} $.
\end{enumerate}
\end{lemma}
\begin{proof}
  1)   It is clear that $\{h : 1 \leq h \leq 2^{\frac{s}{2}} - 2\} \subseteq A \subseteq D_{(0,s)}^\prime$.
  
  2) The proof follows the same argument as Lemma \ref{lem-8} and is therefore omitted.
  
  3) Lemma \ref{lem-14} implies that 
$$
D_{(1,s)}^\prime = \left[\left(R_{(0,s)} \setminus B_{(0,s)}\right) \cup A\right]\setminus \left(C_{5 \cdot 2^\frac{s}{2} +5}^{(2,n)} \cup C_{3 \cdot 2^{\frac{s}{2}}-13}^{(2,n)}\right).
$$
Then
 \begin{eqnarray*}
 |D_{(1,s)}^\prime | &=& \left|[(R_{(0,s)} \setminus B_{(0,s)}) \cup A] \right|- \left|C_{5 \cdot 2^\frac{s}{2} +5}^{(2,n)} \cup C_{3 \cdot 2^{\frac{s}{2}}-13}^{(2,n)}\right|\\
 &=&\left|D_{(1,s)}\right|-\left|C_{5 \cdot 2^\frac{s}{2} +5}^{(2,n)}\right|-\left|C_{3 \cdot 2^{\frac{s}{2}}-13}^{(2,n)}\right|\\
 &=&2^{s-1} - 2 
 \end{eqnarray*} 
 and $|D_{(0,s)}^\prime| =n-1-|D_{(1,s)}^\prime |= 2^{s-1} $
 by Lemmas \ref{lem-8} and \ref{lem-14}.
\end{proof}

\begin{theorem}\label{th-16}
Let $s \geq 10$ and $s \equiv 2 \pmod{4}$. Then we derive the following results.
\begin{enumerate}

\item The binary cyclic code $\mathcal{C}_{(1,s)}^\prime$ has parameters
    $
    \left[ 2^s - 1, 2^{s-1} + 1 , d \geq 2^{\frac{s}{2} + 1 } -1 \right],
    $
    and the binary cyclic code $\mathcal{C}_{(0,s)}^\prime$ has parameters
    $
    \left[ 2^s - 1, 2^{s-1} - 1 , d \geq 2^{\frac{s}{2}} - 1 \right].
    $
\item The dual code $\mathcal{C}_{(1,s)}^{\prime \perp}$ has parameters
    $
    \left[ 2^s - 1, 2^{s-1} - 2 , d^\perp \geq 2^{\frac{s}{2}} \right],
    $
    and the dual code $\mathcal{C}_{(0,s)}^{\prime \perp} $ has parameters
    $
    \left[ 2^s - 1, 2^{s-1}, d^\perp \geq 2^{\frac{s}{2}+1}  \right].
    $
\end{enumerate}    
\end{theorem}
\begin{proof}
By 1) and 2) of Lemma \ref{lem-15}, the defining set $D_{(1,s)}^\prime$ of the code  $\mathcal{C}_{(1,s)}^\prime$ contains the set $\{h : 1 \leq h \leq  2^{\frac{s}{2}+1} -2\}$, while the defining set $D_{(0,s)}^\prime$ of  $\mathcal{C}_{(0,s)}^\prime$  contains the set $\{n - h : 1 \leq h \leq 2^{\frac{s}{2}} - 2\}$. The BCH bound implies that
$
d(\mathcal{C}_{(1,s)}^\prime) \geq 2^{\frac{s}{2} + 1 } - 1 \quad \text{and} \quad d(\mathcal{C}_{(0,s)}^\prime) \geq 2^{\frac{s}{2}} - 1,
$
respectively. For $i \in \{0,1\}$, we have $\dim(\mathcal{C}_{(i,s)}^\prime) = n - |D_{(i,s)}^\prime|$, then the dimension of $\mathcal{C}_{(i,s)}^\prime$ can be determined from Lemma \ref{lem-15}. Since $\dim(\mathcal{C}_{(i,s)}^\prime) +\dim(\mathcal{C}'^{\perp}_{(i,s)})= n$, the dimension of $\mathcal{C}_{(i,s)}'^\perp$ directly follows. The minimum distance of $\mathcal{C}_{(i,s)}'^\perp$ is estimated by considering the following cases.

Case 1: If $i =1$. The set $\{h : 0 \leq h \leq 2^{\frac{s}{2}} - 2\} \subseteq \mathbb{Z}_{n} \setminus (-D_{(1,s)}^\prime)$ by 2) of Lemma \ref{lem-15}.
 By the BCH bound,  the desired lower bound on $d(\mathcal{C}'^{\perp}_{(1,s)})$ is obtained.
 
Case 2: If $i =0$. The set $\{n-h : 0 \leq h \leq 2^{\frac{s}{2}+1} - 2\} \subseteq \mathbb{Z}_{n} \setminus (-D_{(0,s)}^\prime)$ By 1) of Lemma \ref{lem-15}.
By the BCH bound,  the desired lower bound on $d(\mathcal{C}'^{\perp}_{(0,s)})$ is obtained.
\end{proof}

\begin{remark}
In comparison with the binary cyclic codes in \cite{CTCD} with the same lengths and dimensions, the cyclic codes $\mathcal{C}_{(1,s)}^\prime$ and $\mathcal{C}_{(0,s)}^\prime$ in Theorem \ref{th-16} have larger lower bounds on the minimum distances and dual distances.
Compared to the binary cyclic codes in \cite{LWQ} with the same lengths and dimensions, the cyclic codes $\mathcal{C}_{(1,s)}^\prime$ and $\mathcal{C}_{(0,s)}^\prime$ in Theorem \ref{th-16} have larger lower bounds on the minimum distances and the same lower bounds on the dual distances.
\end{remark}

\begin{theorem}\label{th-17}
Let $s \equiv 2 \pmod{4}$ and $s \geq 10$. Then the extended code $\overline{\mathcal{C}_{(1,s)}^\prime}$ has parameters
    $$
    \left[ 2^s, 2^{s-1} + 1, d \geq 2^{\frac{s}{2} + 1 } \right],
    $$
    and the extended code $\overline{\mathcal{C}_{(0,s)}^\prime}$ has parameters
    $$
    \left[ 2^s, 2^{s-1} - 1, d \geq 2^{\frac{s}{2}}  \right].
    $$
\end{theorem}
\begin{proof}
The desired results can be deduced from Theorem \ref{th-16} and the definition of the extended codes.
\end{proof}

\subsection{The third construction of  binary cyclic codes with $ d \cdot d^\perp \approx 2n$}
 Let $n = 2^s - 1$ and $s \geq 6$ be even. we have
\[
R_{(i,s)} = \{h : 1 \leq h \leq n-1, \, \omega_2(h) \equiv i \pmod{2}\}
\]
for  $i \in \{0,1\}$. Denote
\[
M := \bigcup_{h=1}^{2^{\frac{s}{2}+1}-4} C_{h}^{(2,n)}
\]
and
\[
N_{(i,s)} := \bigcup_{\substack{h=1 \\ \omega_2(h) \equiv i \pmod{2}}}^{2^{\frac{s}{2}+1}-4} C_{n-h}^{(2,n)}
\]
for each $i \in \{0,1\}$. Let
\[
V_{(1,s)} = 
\begin{cases}
(R_{(0,s)} \setminus N_{(0,s)}) \cup M & \text{if } s \equiv 2 \pmod{4}, \\
(R_{(1,s)} \setminus N_{(1,s)}) \cup M & \text{if } s \equiv 0 \pmod{4},
\end{cases}
\]
and $V_{(0,s)} = \{1,2,\cdots,n-1\} \setminus V_{(1,s)}$. Then both $V_{(0,s)}$ and $V_{(1,s)}$ consist of unions of certain $2$-cyclotomic cosets. 
 Let $\mathcal{C}_{(i,s)}''$ be the binary cyclic code of length $n$ with defining set  $V_{(i,s)}$ for  $i \in \{0,1\}$. The following lemmas are used to characterize the parameters of $\mathcal{C}_{(i,s)}''$.

\begin{lemma}\label{lem-21}
Let $s \geq 6$ be even and $i$ be an integer such that $1 \leq i \leq  2^{\frac{s}{2} + 1} - 4$. Then
\[
 \left(\bigcup_{i=1}^{2^{\frac{s}{2} + 1} - 4} C_i^{(2,n)}\right) \bigcap  \left(\bigcup_{h=1}^{2^{\frac{s}{2} + 1} - 4} C_{n-h}^{(2,n)}\right) = C_{2^{\frac{s}{2}}-1}^{(2,n)}.
\]
\end{lemma}

\begin{proof}
The proof follows the same argument as Lemma \ref{lem-5} and is therefore omitted.
\end{proof}

 \begin{lemma}\label{lem-22}
   Let $s \geq 6$ be even. The following results are obtained.
   \begin{enumerate}
    \item  $N_{(i,s)} \subseteq R_{(i,s)}$ for $i \in \{0, 1\}$.
    
    \item If $s \equiv 2 \pmod{4}$, then $N_{(0,s)} \cap M = \emptyset$ and $N_{(1,s)} \cap M = C_{2^{\frac{s}{2}}-1}^{(2,n)}$.
    
    \item If $s \equiv 0 \pmod{4}$, then $N_{(0,s)} \cap M = C_{2^{\frac{s}{2}}-1}^{(2,n)}$ and $N_{(1,s)} \cap M = \emptyset$.
\end{enumerate}
 \end{lemma}
 \begin{proof}
 The proof follows the same argument as Lemma \ref{lem-7} and is therefore omitted.
 \end{proof}
  \begin{lemma}\label{lem-23}
Let $s \geq 6$ be even. Then the following hold.
\begin{enumerate}
    \item $\{h : 1 \leq h \leq  2^{\frac{s}{2}+1} -4\} \subseteq V_{(1,s)}$.
    
    \item $\{n - h : 1 \leq h \leq 2^{\frac{s}{2}} - 2\} \subseteq V_{(0,s)}$.
    
    \item If $s \equiv 2 \pmod{4}$, then $|V_{(1,s)}| = 2^{s-1} - 2 + \frac{s}{2}$ and $|V_{(0,s)}| = 2^{s-1} - \frac{s}{2}$.
    
    \item If $s \equiv 0 \pmod{4}$, then $|V_{(1,s)}| = 2^{s-1} - \frac{s}{2}$ and $|V_{(0,s)}| = 2^{s-1} - 2 + \frac{s}{2}$.
\end{enumerate}
\end{lemma}

 \begin{proof}
 The proof follows the same argument as Lemma \ref{lem-8} and is therefore omitted.
 \end{proof}

For $i \in \{0,1\}$, the following theorem characterizes the parameters of $\mathcal{C}_{(i,s)}''$. 
 
\begin{theorem}\label{th-22}
Let $s \geq 6$ be even. Then the following conclusions are obtained.
\begin{enumerate}
    \item If $s \equiv 2 \pmod{4}$, then the code $\mathcal{C}_{(1,s)}''$ has parameters
    $
    \left[ 2^s - 1, 2^{s-1} + 1 - \frac{s}{2}, d \geq 2^{\frac{s}{2} + 1 } - 3 \right],
    $
    and the code $\mathcal{C}_{(0,s)}''$ has parameters
    $
    \left[ 2^s - 1, 2^{s-1} - 1 + \frac{s}{2}, d \geq 2^{\frac{s}{2}} - 1 \right].
    $
    
    \item If $s \equiv 0 \pmod{4}$, then the code $\mathcal{C}_{(1,s)}''$ has parameters
    $
    \left[ 2^s - 1, 2^{s-1} - 1 + \frac{s}{2}, d \geq  2^{\frac{s}{2} + 1} - 3 \right],
    $
    and the code $\mathcal{C}_{(0,s)}''$ has parameters
    $
    \left[ 2^s - 1, 2^{s-1} + 1 - \frac{s}{2}, d \geq 2^{\frac{s}{2}} - 1 \right].
    $
\end{enumerate}
\end{theorem}

\begin{proof}
By 1) and 2) of Lemma \ref{lem-23}, the defining set $V_{(1,s)}$ of the code $\mathcal{C}_{(1,s)}''$ contains the set $\{h : 1 \leq h \leq  2^{\frac{s}{2}+1} -4\}$, while the defining set $V_{(0,s)}$ of  $\mathcal{C}_{(0,s)}''$ contains the set $\{n - h : 1 \leq h \leq 2^{\frac{s}{2}} - 2\}$. The BCH bound implies that
$
d(C_{(1,s)}) \geq 2^{\frac{s}{2} + 1 } - 1 \quad \text{and} \quad d(C_{(0,s)}) \geq 2^{\frac{s}{2}} - 1,
$
respectively. For each $i \in \{0,1\}$, we have $\dim(\mathcal{C}_{(i,s)}) = n - |V_{(i,s)}|$. Then the dimensions directly follow from Lemma \ref{lem-23}. 
\end{proof}

\begin{theorem}\label{th-23}
Let $s \geq 6$ be even. Then the following conclusions are obtained.
\begin{enumerate}
    \item If $s \equiv 2 \pmod{4}$, then the extended code $\overline{\mathcal{C}_{(1,s)}''}$ has parameters
    $
    \left[ 2^s, 2^{s-1} + 1 - \frac{s}{2}, d \geq 2^{\frac{s}{2} + 1 } - 2 \right],
    $
    and the extended code $\overline{\mathcal{C}_{(0,s)}''}$ has parameters
    $
    \left[ 2^s, 2^{s-1} - 1 + \frac{s}{2}, d \geq 2^{\frac{s}{2}}  \right].
    $
    
    \item If $s \equiv 0 \pmod{4}$, then the extended code $\overline{\mathcal{C}_{(1,s)}''}$ has parameters
    $
    \left[ 2^s, 2^{s-1} - 1 + \frac{s}{2}, d \geq  2^{\frac{s}{2} + 1} - 2 \right],
    $
    and the extended code $\overline{\mathcal{C}_{(0,s)}''}$ has parameters
    $
    \left[ 2^s, 2^{s-1} + 1 - \frac{s}{2}, d \geq 2^{\frac{s}{2}}  \right].
    $
\end{enumerate}
\end{theorem}

\begin{proof}
The desired conclusions can be deduced from Theorem \ref{th-22} and the definition of the extended codes.
\end{proof}

\begin{theorem}\label{th-24}
Let $s \geq 6$ be even. Then the following conclusions are obtained.
\begin{enumerate}
    \item If $s \equiv 2 \pmod{4}$, then the dual code $\mathcal{C}_{(1,s)}''^\perp$ has parameters
    $
    \left[ 2^s - 1, 2^{s-1} + \frac{s}{2} -2, d^\perp \geq 2^{\frac{s}{2}}  \right],
    $
    and the dual code $\mathcal{C}_{(0,s)}''^\perp$ has parameters
    $
    \left[ 2^s - 1, 2^{s-1} - \frac{s}{2}, d^\perp \geq 2^{\frac{s}{2}+1} - 2 \right].
    $
    
    \item If $s \equiv 0 \pmod{4}$, then the dual code $\mathcal{C}_{(1,s)}''^\perp$ has parameters
    $
    \left[ 2^s - 1, 2^{s-1} - \frac{s}{2}, d^\perp \geq  2^{\frac{s}{2}} \right],
    $
    and the dual code $\mathcal{C}_{(0,s)}''^\perp$ has parameters
    $
    \left[ 2^s - 1, 2^{s-1} + \frac{s}{2} -2, d^\perp \geq 2^{\frac{s}{2}+1} - 2 \right].
    $
\end{enumerate}
\end{theorem}

\begin{proof}
Observe that $\dim(\mathcal{C}_{(i,s)}'') = n - \dim(\mathcal{C}''^{\perp}_{(i,s)})$. Then the  dimension of $\dim(\mathcal{C}_{(i,s)}''^\perp)$ for each $i \in \{0,1\}$ follows from Theorem \ref{th-22}. The minimum distance of $\mathcal{C}''^{\perp}_{(i,s)}$ is estimated by considering the following cases. 

Case 1: Let $i =1$. The set $\{h : 0 \leq h \leq 2^{\frac{s}{2}} - 2\} \subseteq \mathbb{Z}_{n} \setminus (-V_{(1,s)})$ by 2) of Lemma \ref{lem-23}.
 By the BCH bound, the desired lower bound on  $d(\mathcal{C}''^{\perp}_{(1,s)})$ is obtained.
 
Case 2: Let $i =0$.  The set $\{n-h : 0 \leq h \leq 2^{\frac{s}{2}+1} - 4\} \subseteq \mathbb{Z}_{n} \setminus (-V_{(0,s)})$ by 1) of Lemma \ref{lem-23}.
 By the BCH bound, the desired lower bound on  $d(\mathcal{C}''^{\perp}_{(0,s)})$ is obtained.
  
\end{proof}

\begin{example}
  If $s = 6$, then $n = 2^6 - 1 = 63$. By Magma, we obtain that
 $\mathcal{C}_{(1,s)}''$  is a binary cyclic code with  parameters $\left[63, 30, 13 \right]$, the extended code $\overline{\mathcal{C}_{(1,s)}''}$ has parameters $\left[64, 30, 14 \right]$, the dual code $\mathcal{C}_{(1,s)}''^\perp$ has parameters $\left[63, 33, 12 \right]$, the extended code $\overline{\mathcal{C}''_{(0,s)}}$ has parameters $\left[64, 34, 12 \right]$ and the dual code $\mathcal{C}_{(0,s)}''^\perp$ has parameters $\left[63, 29, 14 \right]$.  The above five binary cyclic codes have the best known parameters according to the Database in \cite{Grassl}. The binary cyclic code  $\mathcal{C}_{(0,s)}''$ has parameters $\left[63, 34, 11 \right]$ and the minimum distance is very close to that of the best known parameters $[63,34,12]$ for binary linear codes of length $63$ and dimension $34$. 
\end{example}

\begin{remark}
In comparison with the binary cyclic codes in \cite{SUN}, the cyclic code $\mathcal{C}_{(1,s)}''$ in Theorem \ref{th-22} has larger minimum distance and the same dual distance, and
the cyclic codes  $\mathcal{C}_{(0,s)}''$ in Theorem \ref{th-22} has larger dual distance and the same minimum distance.
\end{remark}

\subsection{The fourth construction of  binary cyclic codes with $ d \cdot d^\perp \approx 2n$}
In this subsection, we give the fourth construction of  binary cyclic codes with $ d \cdot d^\perp \approx 2n$ by modifying the third construction. 

Let $s \equiv 2 \pmod{4}$, 
\[
\widetilde{V}_{(1,s)}:= \left[R_{(0,s)} \setminus \left(N_{(0,s)} \cup C_{7 \cdot 2^\frac{s}{2} +7}^{(2,n)} \right)\right] \cup M 
\]
and $\widetilde{V}_{(0,s)} := \{1,2,\cdots,n-1\} \setminus \widetilde{V}_{(1,s)}$.
 Then both $\widetilde{V}_{(0,s)}$ and $\widetilde{V}_{(1,s)}$ consist of unions of certain $2$-cyclotomic cosets. Let $\widetilde{\mathcal{C}}_{(i,s)}$ be the binary cyclic code of length $n$ with defining set $\widetilde{V}_{(i,s)}$ for  $i \in \{0,1\}$. The following lemmas are used to characterize the parameters of $\widetilde{\mathcal{C}}_{(i,s)}$.

\begin{lemma}\label{lem-25}
  Let $s \equiv 2 \pmod{4}$ and $s\geq 10$. The following conclusions are obtained.
  \begin{enumerate}
    \item $C_{7 \cdot 2^\frac{s}{2} +7}^{(2,n)} \subseteq R_{(0,s)}$.
    \item The coset leader of $C_{7 \cdot 2^\frac{s}{2} +7}^{(2,n)}$ is $7 \cdot 2^\frac{s}{2} + 7$ and $\left|C_{7 \cdot 2^\frac{s}{2} + 7}^{(2,n)}\right| = \frac{s}{2}$.
    \item $C_{7 \cdot 2^\frac{s}{2} + 7}^{(2,n)} \cap M=\emptyset$.
    \item $C_{7 \cdot 2^\frac{s}{2} + 7}^{(2,n)} \cap N_{(0,s)}=\emptyset$.
  \end{enumerate}
\end{lemma}

\begin{proof}
The proof follows the same argument as Lemma \ref{lem-14} and is therefore omitted.
\end{proof}

\begin{lemma}\label{lem-26}
Let $s \equiv 2 \pmod 4$ and $s \geq 10$. Then the following hold.
\begin{enumerate}
    \item $\{h : 1 \leq h \leq  2^{\frac{s}{2}+1} -4\} \subseteq \widetilde{V}_{(1,s)}$.
    
    \item $\{n - h : 1 \leq h \leq 2^{\frac{s}{2}} - 2\} \subseteq \widetilde{V}_{(0,s)}$.
    
    \item $|\widetilde{V}_{(1,s)}| = 2^{s-1} - 2 $ and $|\widetilde{V}_{(0,s)}| = 2^{s-1} $.
\end{enumerate}
\end{lemma}
\begin{proof}
The proof follows the same argument as Lemma \ref{lem-15} and is therefore omitted.
\end{proof}

\begin{theorem}\label{th-29}
Let $s \geq 10$ and $s \equiv 2 \pmod{4}$. Then we obtain the following conclusions.
\begin{enumerate}

\item The binary cyclic code $\widetilde{\mathcal{C}}_{(1,s)}$ has parameters
    $
    \left[ 2^s - 1, 2^{s-1} + 1 , d \geq 2^{\frac{s}{2} + 1 } -3 \right]
    $
    and the binary cyclic code $\widetilde{\mathcal{C}}_{(0,s)}$ has parameters
    $
    \left[ 2^s - 1, 2^{s-1} - 1 , d \geq 2^{\frac{s}{2}} - 1 \right].
    $
\item The dual code $\widetilde{\mathcal{C}}_{(1,s)}^{ \perp}$ has parameters
    $
    \left[ 2^s - 1, 2^{s-1} - 2 , d^\perp \geq 2^{\frac{s}{2}} \right]
    $
    and the dual code $\widetilde{\mathcal{C}}_{(0,s)}^{ \perp} $ has parameters
    $
    \left[ 2^s - 1, 2^{s-1}, d^\perp \geq 2^{\frac{s}{2}+1} -2 \right].
    $
\end{enumerate}    
\end{theorem}

\begin{proof}
By 1) and 2) of Lemma \ref{lem-26}, the defining set $\widetilde{V}_{(1,s)}$ of the code $\widetilde{\mathcal{C}}_{(1,s)}$ contains the set $\{h : 1 \leq h \leq  2^{\frac{s}{2}+1} -4\}$, while the defining set $\widetilde{V}_{(0,s)}$ of  $\widetilde{\mathcal{C}}_{(0,s)}$ contains the set $\{n - h : 1 \leq h \leq 2^{\frac{s}{2}} - 2\}$.
The BCH bound implies that
$
d(\widetilde{C}_{(1,s)}) \geq 2^{\frac{s}{2} + 1 } +5 \quad \text{and} \quad d(\widetilde{C}_{(0,s)}) \geq 2^{\frac{s}{2}} - 1,
$
respectively.
 For each $i \in \{0,1\}$, we have  $\dim(\widetilde{\mathcal{C}}_{(i,s)}) = n - |\widetilde{V}_{(i,s)}|$. Then the dimensions follow from Lemma \ref{lem-26}. Since $\dim(\widetilde{\mathcal{C}}_{(i,s)}) +\dim(\widetilde{\mathcal{C}}^{\perp}_{(i,s)})= n$, the dimension of $\widetilde{\mathcal{C}}^{\perp}_{(i,s)}$ for each $i \in \{0,1\}$ directly follows. To study the minimum distances of the dual codes, we consider the following two cases. 

Case 1: Let $i =1$. The set $\{h : 0 \leq h \leq 2^{\frac{s}{2}} - 2\} \subseteq \mathbb{Z}_{n} \setminus (-\widetilde{V}_{(1,s)})$ by 2) of Lemma \ref{lem-26}.
  By the BCH bound, the desired lower bound on $d(\widetilde{\mathcal{C}}^{\perp}_{(1,s)})$ is obtained.
 
Case 2: Let $i =0$. The set $\{n-h : 0 \leq h \leq 2^{\frac{s}{2}+1} + 4\} \subseteq \mathbb{Z}_{n} \setminus (-\widetilde{V}_{(0,s)})$ by 1) of Lemma \ref{lem-26}.
   By the BCH bound, the desired lower bound on   $d(\widetilde{\mathcal{C}}^{\perp}_{(0,s)})$ is obtained.
 
 Then we have completed the proof. 
\end{proof}

\begin{remark}
In comparison with the binary cyclic codes with the same lengths and dimensions in \cite{CTCD}, the cyclic codes $\widetilde{\mathcal{C}}_{(1,s)}$ and $\widetilde{\mathcal{C}}_{(0,s)}$ in Theorem \ref{th-29} have larger minimum distances and dual distances. 
In comparison with the binary cyclic codes with the same lengths and dimensions in \cite{LWQ}, the cyclic codes $\widetilde{\mathcal{C}}_{(1,s)}$ and $\widetilde{\mathcal{C}}_{(0,s)}$ in Theorem \ref{th-29} have larger minimum distances and the same dual distances.
\end{remark}

The extended code of $\widetilde{\mathcal{C}}_{(i,s)}$ are characterized as follows.

\begin{theorem}\label{th-30}
Let $s \equiv 2 \pmod{4}$ and $s \geq 10$. Then the extended code $\overline{\widetilde{\mathcal{C}}_{(1,s)}}$ has parameters
    $$
    \left[ 2^s, 2^{s-1} + 1, d \geq 2^{\frac{s}{2} + 1 } -2 \right]
    $$
    and the extended code $\overline{\widetilde{\mathcal{C}}_{(0,s)}}$ has parameters
    $$
    \left[ 2^s, 2^{s-1} - 1, d \geq 2^{\frac{s}{2}}  \right].
    $$
\end{theorem}

\begin{proof}
 The desired results can be deduced from Theorem \ref{th-29} and the definition of the extended codes.
\end{proof}

\subsection{The fifth construction of binary cyclic codes with $d\cdot d^\perp \approx n$}
 Let $n = 2^s - 1$ and $s \geq 6$ be even. Let
\[
R_{(i,s)} := \{h : 1 \leq h \leq n-1, \, \omega_2(h) \equiv i \pmod{2}\}
\]
for $i \in \{0,1\}$. Let
\[
E := \bigcup_{h=1}^{2^{\frac{s}{2}+1}-1} C_{h}^{(2,n)}
\]
and
\[
F_{(i,s)} := \bigcup_{\substack{h=1 \\ \omega_2(h) \equiv i \pmod{2}}}^{2^{\frac{s}{2}+1}-1} C_{n-h}^{(2,n)}
\]
for each $i \in \{0,1\}$. Let
\[
U_{(1,s)} = 
\begin{cases}
(R_{(0,s)} \setminus F_{(0,s)}) \cup E & \text{if } s \equiv 2 \pmod{4}, \\
(R_{(1,s)} \setminus F_{(1,s)}) \cup E & \text{if } s \equiv 0 \pmod{4},
\end{cases}
\]
and $U_{(0,s)} = \{1,2,\cdots,n-1\} \setminus U_{(1,s)}$. 
Then both $U_{(0,s)}$ and $U_{(1,s)}$ consist of unions of certain $2$-cyclotomic cosets. 
Let $\mathcal{C}_{(i,s)}'''$ be the binary cyclic code of length $n$ with defining set $U_{(i,s)}$ for  $i \in \{0,1\}$. The following lemmas are used to characterize the parameters of $\mathcal{C}_{(i,s)}'''$.

\begin{lemma}\label{lem-28}
Let $s \geq 6$ be even and $i$ be an integer with $1 \leq i \leq  2^{\frac{s}{2} + 1} - 1$. Then we have
\begin{align*}
 &\left(\bigcup_{i=1}^{2^{\frac{s}{2} + 1} - 1} C_i^{(2,n)}\right) \bigcap  \left(\bigcup_{h=1}^{2^{\frac{s}{2} + 1} - 1} C_{n-h}^{(2,n)}\right)\\
 &= C_{2^{\frac{s}{2}-1}-1}^{(2,n)} \cup C_{2^{\frac{s}{2}}-1}^{(2,n)} \cup C_{3 \cdot 2^{\frac{s}{2}-1}-1}^{(2,n)} \cup C_{2^{\frac{s}{2}+1}-3}^{(2,n)} \cup C_{2^{\frac{s}{2}+1}-1}^{(2,n)}.
\end{align*}
\end{lemma}
\begin{proof}
For $1 \leq i \leq  2^{\frac{s}{2} + 1} - 1$, we have $C_i^{(2,n)} \subseteq \bigcup_{h=1}^{2^{\frac{s}{2} + 1} - 2} C_{n-h}^{(2,n)}$ if and only if $i \in C_{n-h}^{(2,n)}$ for some $1 \leq h \leq  2^{\frac{s}{2} + 1} - 1$. Equivalently, one can find an integer $k$ with $0 \leq k \leq s - 1$ such that
\begin{eqnarray}\label{eqn-11}
i + h \cdot 2^k \equiv 0 \pmod{n}.
\end{eqnarray}
 
Case 1: If $0 \leq k \leq \frac{s}{2}-2$, then for any $1 \leq i, h \leq 2^{\frac{s}{2} + 1} - 1$ it holds that
\begin{eqnarray*}
1 < i + h \cdot 2^k &\leq (2^{\frac{s}{2}+1}-1)(2^{\frac{s}{2}-2}+1)\\
&=2^{s-1} + 7 \cdot 2^{\frac{s}{2}-2} -1 < n,
\end{eqnarray*}
which implies that Equation \eqref{eqn-11} has no solution.

Case 2: If $k = \frac{s}{2}-1$, then for any $1 \leq i, h \leq 2^{\frac{s}{2} + 1} - 1$ we have
\begin{eqnarray*}
1 < i + h \cdot 2^k &\leq (2^{\frac{s}{2}+1}-1)(2^{\frac{s}{2}-1}+1)\\
&= 2^s + 3 \cdot 2^{\frac{s}{2}-1} -1 <2n.
\end{eqnarray*}
In this case, Equation \eqref{eqn-11} holds if and only if $i + h \cdot 2^{\frac{s}{2}-1} = 2^s - 1$. Therefore, the equation has the following solutions: $(i, h) = (2^{\frac{s}{2}} - 1, 2^{\frac{s}{2}+1} - 2)$, $(2^{\frac{s}{2}-1} - 1, 2^{\frac{s}{2}+1} - 1)$, $(2^{\frac{s}{2}+1} - 1, 2^{\frac{s}{2}+1} - 4)$ and $(3 \cdot 2^{\frac{s}{2}-1} - 1, 2^{\frac{s}{2}+1} - 3)$.

Case 3: If $k = \frac{s}{2}$, then for any $1 \leq i, h \leq 2^{\frac{s}{2} + 1} - 1$ we have
\begin{eqnarray*}
1 < i + h \cdot 2^k &\leq (2^{\frac{s}{2}+1}-1)(2^\frac{s}{2}+1)\\
&= 2^{s+1} + 2^\frac{s}{2}-1  < 3n.
\end{eqnarray*}
In this case, if $i + h \cdot 2^{\frac{s}{2}} = 2^s - 1$, then the solutions are given by $(i, h) = (2^{\frac{s}{2}} - 1, 2^{\frac{s}{2}} - 1)$ and $(2^{\frac{s}{2}+1} - 1, 2^{\frac{s}{2}} - 2)$. If $i + h \cdot 2^{\frac{s}{2}} = 2^{s+1} - 2$, then this equation  has no solution.

Case 4: If $k = \frac{s}{2}+1$, then $s - k = \frac{s}{2} - 1$. Similarly to Case 2, Equation \eqref{eqn-11} has solutions $(i, h) = ( 2^{\frac{s}{2}+1} - 2, 2^{\frac{s}{2}} - 1)$, $( 2^{\frac{s}{2}+1} - 1, 2^{\frac{s}{2}-1} - 1)$, $( 2^{\frac{s}{2}+1} - 4, 2^{\frac{s}{2}+1} - 1)$ and $( 2^{\frac{s}{2}+1} - 3, 3 \cdot 2^{\frac{s}{2}-1} - 1)$.

Case 5: If $ \frac{s}{2} + 2 \leq k \leq s-1$, then $1 \leq s-k \leq \frac{s}{2}-2$. Similarly to Case 1, Equation \eqref{eqn-11} has no solution.

Therefore, the desired conclusion is obtained from the above discussion.
\end{proof}

 \begin{lemma}\label{lem-29}
   Let $s \geq 6$ be even. Then the following conclusions are obtained.
   \begin{enumerate}
    \item  $F_{(i,s)} \subseteq R_{(i,s)}$ for  $i \in \{0, 1\}$.
    
    \item If $s \equiv 2 \pmod{4}$, then $F_{(0,s)} \cap E = C_{2^{\frac{s}{2}-1}-1}^{(2,n)} \cup C_{2^{\frac{s}{2}+1}-1}^{(2,n)}$ and $F_{(1,s)} \cap E = C_{2^{\frac{s}{2}}-1}^{(2,n)} \cup C_{3 \cdot 2^{\frac{s}{2}-1}-1}^{(2,n)} \cup C_{2^{\frac{s}{2}+1}-3}^{(2,n)}$.
    
    \item If $s \equiv 0 \pmod{4}$, then $F_{(0,s)} \cap E = C_{2^{\frac{s}{2}}-1}^{(2,n)} \cup C_{3 \cdot 2^{\frac{s}{2}-1}-1}^{(2,n)} \cup C_{2^{\frac{s}{2}+1}-3}^{(2,n)}$ and $F_{(1,s)} \cap E = C_{2^{\frac{s}{2}-1}-1}^{(2,n)} \cup C_{2^{\frac{s}{2}+1}-1}^{(2,n)}$.
\end{enumerate}

 \end{lemma}
\begin{proof}
The proof follows the same argument as Lemma \ref{lem-7} and is therefore omitted.
\end{proof}

 \begin{lemma}\label{lem-30}
Let $s \geq 6$ be even. Then the following conclusions are obtained.
\begin{enumerate}
    \item $\{h : 1 \leq h \leq  2^{\frac{s}{2}+1} + 4\} \subseteq U_{(1,s)}$.
    
    \item $\{n - h : 1 \leq h \leq 2^{\frac{s}{2}-1} - 2\} \subseteq U_{(0,s)}$.
    
    \item If $s \equiv 2 \pmod{4}$, then $|U_{(1,s)}| = 2^{s-1} - 2 + \frac{5s}{2}$ and $|U_{(0,s)}| = 2^{s-1} - \frac{5s}{2}$.
    
    \item If $s \equiv 0 \pmod{4}$, then $|U_{(1,s)}| = 2^{s-1} + \frac{3s}{2}$ and $|U_{(0,s)}| = 2^{s-1} - 2 - \frac{3s}{2}$.
\end{enumerate}
\end{lemma}
\begin{proof}
The proof follows the same argument as Lemma \ref{lem-8} and is therefore omitted.
\end{proof}

The following theorem characterizes the parameters of $\mathcal{C}'''_{(i,s)}$. 
 
\begin{theorem}\label{th-31}
Let $s \geq 6$ be even. Then the following conclusions are obtained.
\begin{enumerate}
    \item If $s \equiv 2 \pmod{4}$, then the code $\mathcal{C}'''_{(1,s)}$ has parameters
    $
    \left[ 2^s - 1, 2^{s-1} + 1 - \frac{5s}{2}, d \geq 2^{\frac{s}{2} + 1 } + 5 \right]
    $
    and the  code $\mathcal{C}'''_{(0,s)}$ has parameters
    $
    \left[ 2^s - 1, 2^{s-1} - 1 + \frac{5s}{2}, d \geq 2^{\frac{s}{2}-1} - 1 \right].
    $
    
    \item If $s \equiv 0 \pmod{4}$, then the code $\mathcal{C}'''_{(1,s)}$ has parameters
    $
    \left[ 2^s - 1, 2^{s-1} - 1 - \frac{3s}{2}, d \geq  2^{\frac{s}{2} + 1} +5 \right]
    $
    and the code $\mathcal{C}'''_{(0,s)}$ has parameters
    $
    \left[ 2^s - 1, 2^{s-1} + 1 + \frac{3s}{2}, d \geq 2^{\frac{s}{2}-1} - 1 \right].
    $
\end{enumerate}
\end{theorem}

\begin{proof}
By 1) and 2) of Lemma \ref{lem-30}, the defining set $U_{(1,s)}$ of the code $\mathcal{C}_{(1,s)}'''$ contains the set $\{h : 1 \leq h \leq  2^{\frac{s}{2}+1} + 4\}$, while the defining set $U_{(0,s)}$ of  $\mathcal{C}_{(0,s)}'''$  contains the set $\{n - h : 1 \leq h \leq 2^{\frac{s}{2}-1} - 2\}$.
The BCH bound implies that
$
d(\mathcal{C}'''_{(1,s)}) \geq 2^{\frac{s}{2} + 1 } +5 \quad \text{and} \quad d(\mathcal{C}'''_{(0,s)}) \geq 2^{\frac{s}{2}} - 1,
$
respectively.
 We have $\dim(\mathcal{C}'''_{(i,s)}) = n - |U_{(i,s)}|$ for $i \in \{0,1\}$. Then the dimensions follow from Lemma \ref{lem-30}. 
\end{proof}

\begin{theorem}\label{th-32}
Let $s \geq 6$ be even. Then the following conclusions are obtained.
\begin{enumerate}
    \item If $s \equiv 2 \pmod{4}$, then the extended code $\overline{\mathcal{C}'''_{(1,s)}}$ has parameters
    $
    \left[ 2^s, 2^{s-1} + 1 - \frac{5s}{2}, d \geq 2^{\frac{s}{2} + 1 } + 6 \right]
    $
    and the extended code $\overline{\mathcal{C}'''_{(0,s)}}$ has parameters
    $
    \left[ 2^s, 2^{s-1} - 1 + \frac{5s}{2}, d \geq 2^{\frac{s}{2}-1}\right].
    $
    
    \item If $s \equiv 0 \pmod{4}$, then the extended code $\overline{\mathcal{C}'''_{(1,s)}}$ has parameters
    $
    \left[ 2^s, 2^{s-1} - 1 - \frac{3s}{2}, d \geq  2^{\frac{s}{2} + 1} +6 \right]
    $
    and the extended code $\overline{\mathcal{C}'''_{(0,s)}}$ has parameters
   $
    \left[ 2^s, 2^{s-1} + 1 + \frac{3s}{2}, d \geq 2^{\frac{s}{2}-1}  \right].
   $
\end{enumerate}
\end{theorem}

\begin{proof}
 The desired results can be deduced from Theorem \ref{th-31} and the definition of the extended codes.
\end{proof}

\begin{theorem}\label{th-33}
Let $s \geq 6$ be even. Then the following conclusions are obtained.
\begin{enumerate}
    \item If $s \equiv 2 \pmod{4}$, then the dual code $\mathcal{C}_{(1,s)}^{'''\perp}$ has parameters
    $
    \left[ 2^s - 1, 2^{s-1} - 2 + \frac{5s}{2}, d^\perp \geq 2^{\frac{s}{2}-1} \right]
    $
    and the dual code $\mathcal{C}_{(0,s)}^{'''\perp}$ has parameters
    $
    \left[ 2^s - 1, 2^{s-1} - \frac{5s}{2}, d^\perp \geq 2^{\frac{s}{2}} + 6 \right].
    $
    
    \item If $s \equiv 0 \pmod{4}$, then the dual code $\mathcal{C}_{(1,s)}^{'''\perp}$ has parameters
    $
    \left[ 2^s - 1, 2^{s-1} + \frac{3s}{2}, d^\perp \geq  2^{\frac{s}{2}-1}  \right]
    $
    and the dual code $\mathcal{C}_{(0,s)}^{'''\perp}$ has parameters
    $
    \left[ 2^s - 1, 2^{s-1} -2 - \frac{3s}{2}, d^\perp \geq 2^{\frac{s}{2}+1} + 6 \right].
    $
\end{enumerate}
\end{theorem}

\begin{proof}
Since $\dim(\mathcal{C}'''_{(i,s)}) = n - \dim(\mathcal{C}^{'''\perp}_{(i,s)})$ for $i\in\{0,1\}$, the dimension of $\dim(\mathcal{C}_{(i,s)}^{'''\perp})$ follows from Theorem \ref{th-31}. Note that
the defining set of $ \mathcal{C}^{'''\perp}_{(i,s)}$ is $\mathbb{Z}_{n} \setminus (-U_{(i,s)})$. To study the minimum distances of the dual codes, we consider the following two cases. 

Case 1: Let $i =1$. The set $\{h : 0 \leq h \leq 2^{\frac{s}{2}} - 2\} \subseteq \mathbb{Z}_{n} \setminus (-U_{(1,s)})$ by 2) of Lemma \ref{lem-30}.
By the BCH bound,  the desired lower bound on $d(\mathcal{C}^{'''\perp}_{(1,s)})$ is obtained.
 
Case 2: Let $i =0$. The set $\{n-h : 0 \leq h \leq 2^{\frac{s}{2}+1} + 4\} \subseteq \mathbb{Z}_{n} \setminus (-U_{(0,s)})$ by 1) of Lemma \ref{lem-30}.
 By the BCH bound,  the desired lower bound on $d(\mathcal{C}^{'''\perp}_{(0,s)})$ is obtained.
\end{proof}

\begin{example}
We now use Magma to give an example which shows that the lower bounds on the minimum distances of the codes in this subsection are tight. 
  When $s=6$,  the code $\mathcal{C}'''_{(1,s)}$ has parameters $\left[63, 18, 21 \right]$, the code $\mathcal{C}'''_{(0,s)}$ has parameters $\left[63, 46, 7 \right]$, the dual code $\mathcal{C}_{(0,s)}^{'''\perp}$ has parameters $\left[63, 17, 22 \right]$. These three binary codes have the best  known parameters according to the Database\cite{Grassl}. The extended code $\overline{\mathcal{C}'''_{(1,s)}}$ has parameters $\left[64, 18, 22 \right]$, the dual code $\mathcal{C}_{(1,s)}^{'''\perp}$ has parameters $\left[63, 45, 8 \right]$, the extended code $\overline{\mathcal{C}'''_{(0,s)}}$ has parameters $\left[64, 46, 8 \right]$. These three binary codes are optimal according to the Database\cite{Grassl}.
\end{example}

\subsection{The sixth construction of binary cyclic codes with $d\cdot d^\perp \approx n$}
In this subsection, we modify the fifth construction of binary cyclic codes and propose the sixth construction. 

Let $s \equiv 2 \pmod{4}$ and
$$G: = C_{3 \cdot 2^{\frac{s}{2}}+3}^{(2,n)} \cup C_{3 \cdot 2^{\frac{s}{2}}-7}^{(2,n)} \cup C_{3 \cdot 2^{\frac{s}{2}}-11}^{(2,n)}.$$
Let
\[
\widehat{U}_{(1,s)} := \left[R_{(0,s)} \setminus (F_{(0,s)} \cup G)\right] \cup E 
\]
and $\widehat{U}_{(0,s)} := \{1,2,\cdots,n-1\} \setminus \widehat{U}_{(1,s)}$.  Then both $\widehat{U}_{(0,s)}$ and $\widehat{U}_{(1,s)}$ consist of unions of certain $2$-cyclotomic cosets. 
 Let $\widehat{\mathcal{C}}_{(i,s)}$ be the binary cyclic code of length $n$ with defining set $\widehat{U}_{(i,s)}$ for $i \in \{0,1\}$. The following lemmas are used to characterize the parameters of $\widehat{\mathcal{C}}_{(i,s)}$.

\begin{lemma}\label{lem-35}
Let $s \geq 10$ and $s \equiv 2 \pmod{4}$. Then 
\[
G \cap \left(\bigcup_{i=1}^{2^{\frac{s}{2}+1}-1} C^{(2,n)}_{n-i}\right)=\emptyset\text{ and }G\cap F_{(0,s)}=\emptyset.
\]
\end{lemma}

\begin{proof}
It suffices to prove that $3 \cdot 2^{\frac{s}{2}}-j \notin \bigcup_{i=1}^{2^{\frac{s}{2}+1}-2} C^{(2,n)}_{n-i}$ for any $j \in \{-3, 7, 11\}$. Suppose $3 \cdot 2^{\frac{s}{2}}-j \in C^{(2,n)}_{n-i}$ for some $1 \leq i \leq 2^{\frac{s}{2}+1}-2$. One can find an integer $k$ with $0 \leq k \leq s-1$ such that
\begin{align}\label{eq-3}
(3 \cdot 2^{\frac{s}{2}}-j) \cdot 2^k \equiv n-i \pmod{n}.
\end{align}
We now consider the following cases. 

Case 1: If $0 \leq k \leq 2^{\frac{s}{2}-2}$, then
\begin{align*}
(3 \cdot 2^{\frac{s}{2}}-j) \cdot 2^k \pmod{n} &\leq 3 \cdot 2^{s-2} - 2^{\frac{s}{2}-2} \cdot j\\
                                                  &  <  3 \cdot 2^{s-2} + 2^{\frac{s}{2}}\\
                                                  &  < 2^s - 1 - (2^{\frac{s}{2}+1}-1)\\
                                                  & \leq n-i,
\end{align*}
 which contradicts Equation \eqref{eq-3}.
 
 Case 2: If $k = 2^{\frac{s}{2}-1}$, then 
 \begin{align*}
 (3 \cdot 2^{\frac{s}{2}}-j) \cdot 2^k \pmod{n} &= 2^{s-1} + 1 - j \cdot 2^{\frac{s}{2}-1}\\
                                                   & \leq  2^{s-1} + 1 + 3\cdot 2^{\frac{s}{2}} < n-i,
 \end{align*}
  which contradicts Equation \eqref{eq-3}.
  
  Case 3: If $k = 2^{\frac{s}{2}}$, then $(3 \cdot 2^{\frac{s}{2}}-j) \cdot 2^k \pmod{n} = 2^{s} + 2 - j \cdot 2^{\frac{s}{2}} \pmod{n}$.
  Then we consider the following subcases.                                    
\begin{itemize}
  \item If $j=-3$, then $2^{s} + 2 - j \cdot 2^{\frac{s}{2}} \pmod{n} = 3 \cdot 2^\frac{s}{2} + 3 < n-i$.
  \item If $j \in \{7,11\}$, then $2^{s} + 2 - j \cdot 2^{\frac{s}{2}} \leq 2^{s} + 2 - 7 \cdot 2^{\frac{s}{2}} < n-i$.
\end{itemize}
 Therefore $(3 \cdot 2^{\frac{s}{2}}-j) \cdot 2^k \pmod{n} < n-i$ for any $j \in \{-3, 7, 11\}$, which contradicts Equation \eqref{eq-3}.
  
 Case 4: If $\frac{s}{2} + 1\leq k \leq s-4$, we consider the following subcases. 
 \begin{itemize}
 \item If $j \in \{7,11\}$, then
\begin{align*}
& \quad (3 \cdot 2^{\frac{s}{2}}-j) \cdot 2^k \pmod{n} \\
                                                  &= 2^s - 1 + 2^{k + 1-\frac{s}{2}} + 2^{k - \frac{s}{2}} - j \cdot 2^k\\
                                                  &\leq  2^s - 1 + 5 - j \cdot 2^{\frac{s}{2} + 1}\\
                                                  & < 2^s + 4 - 6 \cdot 2^{\frac{s}{2}+1} < n-i,
\end{align*}
\item If $j=-3$, then
\begin{align*}
& \quad (3 \cdot 2^{\frac{s}{2}}-j) \cdot 2^k \pmod{n} \\
                                                  &= 2^{k + 1-\frac{s}{2}} + 2^{k - \frac{s}{2}} + 3 \cdot 2^k\\
                                                  &\leq  3 \cdot 2^{s-4} - 2^{\frac{s}{2}-3} - 2^{\frac{s}{2}-4}< n-i.
\end{align*}
 \end{itemize}
 Therefore $(3 \cdot 2^{\frac{s}{2}}-j) \cdot 2^k \pmod{n} < n-i$ for any $j \in \{-3, 7, 11\}$, which contradicts Equation \eqref{eq-3}.
 
 Case 5: If $k = s-3$, then $ (3 \cdot 2^{\frac{s}{2}}-j) \cdot 2^k \pmod{n} =  2^{\frac{s}{2}-2} + 2^{\frac{s}{2}-3} - j \cdot 2^{s-3} \pmod{n}$.
 Next we consider the following subcases. 
  \begin{itemize}
  \item If $j = -3$, then $2^{\frac{s}{2}-2} + 2^{\frac{s}{2}-3} - j \cdot 2^{s-3} \pmod{n} = 3 \cdot 2^{s-3} + 3 \cdot 2^{\frac{s}{2}-3}$.
   \item If $j = 7$, then $2^{\frac{s}{2}-2} + 2^{\frac{s}{2}-3} - j \cdot 2^{s-3} \pmod{n} = 2^{s-3} + 3 \cdot 2^{\frac{s}{2}-3} -1$.
   \item If $j = 11$, then $2^{\frac{s}{2}-2} + 2^{\frac{s}{2}-3} - j \cdot 2^{s-3} \pmod{n} = 5 \cdot 2^{s-3} + 3 \cdot 2^{\frac{s}{2}-3} -2$.
 \end{itemize}
 Therefore, $(3 \cdot 2^{\frac{s}{2}}-j) \cdot 2^k \pmod{n} \leq 5 \cdot 2^{s-3} + 3 \cdot 2^{\frac{s}{2}-2} -2< n-i $ for any $j \in \{-3, 7, 11\}$,  which contradicts Equation \eqref{eq-3}.
 
 Case 6: If $k = s-2$, then $ (3 \cdot 2^{\frac{s}{2}}-j) \cdot 2^k \pmod{n} =  2^{\frac{s}{2}-1} + 2^{\frac{s}{2}-2} - j \cdot 2^{s-2} \pmod{n}$. We now consider the following subcases. 
 \begin{itemize}
   \item If $j = -3$, then $2^{\frac{s}{2}-1} + 2^{\frac{s}{2}-2} - j \cdot 2^{s-2} \pmod{n} = 3 \cdot 2^{s-2} + 3 \cdot 2^{\frac{s}{2}-3}$.
   \item If $j = 7$, then $2^{\frac{s}{2}-1} + 2^{\frac{s}{2}-2} - j \cdot 2^{s-2} \pmod{n} = 2^{s-2} + 3 \cdot 2^{\frac{s}{2}-2} -2$.
   \item If $j = 11$, then $2^{\frac{s}{2}-1} + 2^{\frac{s}{2}-2} - j \cdot 2^{s-2} \pmod{n} =  2^{s-2} + 3 \cdot 2^{\frac{s}{2}-2} -3$.
 \end{itemize}
 Therefore, $(3 \cdot 2^{\frac{s}{2}}-j) \cdot 2^k \pmod{n} \leq 3 \cdot 2^{s-2} + 3 \cdot 2^{\frac{s}{2}-2} < n-i $ for any $j \in \{-3, 7, 11\}$,  which contradicts Equation \eqref{eq-3}.
 
  Case 7: If $k = s-1$, then $(3 \cdot 2^{\frac{s}{2}}-j) \cdot 2^k \pmod{n} =  2^{\frac{s}{2}} + 2^{\frac{s}{2}-1} - j \cdot 2^{s-1} \pmod{n}$.
  We consider the following subcases. 
 \begin{itemize}
   \item If $j = -3$, then $2^{\frac{s}{2}} + 2^{\frac{s}{2}-1} - j \cdot 2^{s-2} \pmod{n} = 2^{s-1} + 3 \cdot 2^{\frac{s}{2}-1} +1$.
   \item If $j = 7$, then $2^{\frac{s}{2}} + 2^{\frac{s}{2}-1} - j \cdot 2^{s-2} \pmod{n} = 2^{s-1} + 3 \cdot 2^{\frac{s}{2}-1} -4$.
   \item If $j = 11$, then $2^{\frac{s}{2}} + 2^{\frac{s}{2}-1} - j \cdot 2^{s-2} \pmod{n} = 2^{s-1} + 3 \cdot 2^{\frac{s}{2}-1} -6$.
 \end{itemize}
 Therefore, $(3 \cdot 2^{\frac{s}{2}}-j) \cdot 2^k \pmod{n} \leq 2^{s-1} + 3 \cdot 2^{\frac{s}{2}-1} +1 < n-i $ for any $j \in \{-3, 7, 11\}$,  which contradicts Equation \eqref{eq-3}.
 
Therefore, the desired conclusion is obtained from the above discussion.
\end{proof}

  \begin{lemma}\label{lem-36}
  Let $s \equiv 2 \pmod{4}$ and $s \geq 10$. Then the following conclusions are obtained.
  \begin{enumerate}
    \item  $\omega_2(3 \cdot 2^{\frac{s}{2}} - 7)= \omega_2(3 \cdot 2^{\frac{s}{2}} - 11) = \frac{s-2}{2}$ and $G \subseteq R_{(0,s)} $.
    \item $|G| = \frac{5s}{2}$ and $G \cap E=\emptyset$.
  \end{enumerate}
  \end{lemma}
  
  \begin{proof}
    1) It is easy to verify that $3 \cdot 2^{\frac{s}{2}} - 7 = 2^{\frac{s}{2}+1} + \sum_{j=3}^{\frac{s}{2}-1} 2^{j} + 1$ and  $3 \cdot 2^{\frac{s}{2}} - 11 = 2^{\frac{s}{2}+1} + \sum_{j=4}^{\frac{s}{2}-1} 2^{j} + 3$. Then $\omega_2(3 \cdot 2^{\frac{s}{2}} - 7)= \omega_2(3 \cdot 2^{\frac{s}{2}} - 11) = \frac{s-2}{2}$.  Since $s \equiv 2 \pmod{4}$, we have $\omega_2(3 \cdot 2^{\frac{s}{2}} - 7)= \omega_2(3 \cdot 2^{\frac{s}{2}} - 11) \equiv 0 \pmod{2}$. It is clear that $\omega_2(3 \cdot 2^{\frac{s}{2}} + 3) = 4 \equiv 0 \pmod{2}$. Therefore, we have $G \subseteq R_{(0,s)}$.
    
    2) It is easy to see that $3 \cdot 2^{\frac{s}{2}} + 3$, $3 \cdot 2^{\frac{s}{2}} - 7$ and  $3 \cdot 2^{\frac{s}{2}} - 11$ are coset leaders. Since $(3 \cdot 2^{\frac{s}{2}} +3) \cdot 2^\frac{s}{2} \pmod{n} =  3 \cdot 2^{\frac{s}{2}} +3$, we have $\left|C_{3 \cdot 2^{\frac{s}{2}}+3}^{(2,n)}\right|=\frac{s}{2}$.
    Similarly, we have $\left|C_{3 \cdot 2^{\frac{s}{2}}-7}^{(2,n)}\right|=\left| C_{3 \cdot 2^{\frac{s}{2}}-11}^{(2,n)}\right|=s$.
    Then $|G| = \frac{s}{2}+s+s=\frac{5s}{2}$. By Lemma \ref{lem-3},  every integer $i$ with $1 \leq i \leq 2^{\frac{s}{2}+1} -1 $ and $i \equiv 1 \pmod{2}$  is a coset leader. Since $3 \cdot 2^{\frac{s}{2}} + 3 >3 \cdot 2^{\frac{s}{2}} - 7 > 3 \cdot 2^{\frac{s}{2}} - 11 > 2^\frac{s+2}{2} - 1 $, we deduce $C_{3 \cdot 2^{\frac{s}{2}}+3}^{(2,n)}\cap E=\emptyset$. Similarly, we have $C_{3 \cdot 2^{\frac{s}{2}}-7}^{(2,n)}\cap E=\emptyset$ and $C_{3 \cdot 2^{\frac{s}{2}}-11}^{(2,n)}\cap E=\emptyset$.
    Therefore, $G \cap E=\emptyset$ by the definition of $G$. 
  \end{proof}

 \begin{lemma}\label{lem-37}
Let $s \equiv 2 \pmod{4}$ and $s \geq 10$. Then the following conclusions are obtained.
\begin{enumerate}
    \item $\{h : 1 \leq h \leq  2^{\frac{s}{2}+1} + 4\} \subseteq \widehat{U}_{(1,s)}$.
    
    \item $\{n - h : 1 \leq h \leq 2^{\frac{s}{2}-1} - 2\} \subseteq \widehat{U}_{(0,s)}$.
    
    \item $|\widehat{U}_{(1,s)}| = 2^{s-1} - 2 $ and $|\widehat{U}_{(0,s)}| = 2^{s-1}$.
    
\end{enumerate}

\begin{proof}
  1)   It is clear that $\{h : 1 \leq h \leq 2^{\frac{s}{2}} + 4\} \subseteq E \subseteq \widehat{U}_{(1,s)}$.
  
  2) The proof follows the same argument as Lemma \ref{lem-8} and is therefore omitted.
  
  3) By Lemma \ref{lem-35} and \ref{lem-36}, we have
$$
\widehat{U}_{(1,s)} = [(R_{(0,s)} \setminus F_{(0,s)}) \cup E]\setminus G.
$$
Then $|\widehat{U}_{(1,s)}| = \left|[(R_{(0,s)} \setminus F_{(0,s)}) \cup E] \right|- |G|=|U_{1,s}|-|G|=2^{s-1} - 2 $ by Lemmas \ref{lem-30} and \ref{lem-36}.  
Besides, we have $|\widehat{U}_{(0,s)}| = n-1-|\widehat{U}_{(1,s)}| =2^{s-1}$.
\end{proof}
\end{lemma}

The following theorem characterizes the parameters of $\widehat{\mathcal{C}}_{(i,s)}$. 
\begin{theorem}\label{th-38}
Let $s \geq 10$ and $s \equiv 2 \pmod{4}$. Then the following conclusions are obtained.
\begin{enumerate}

\item The binary cyclic code $\widehat{\mathcal{C}}_{(1,s)}$ has parameters
    $
    \left[ 2^s - 1, 2^{s-1} + 1 , d \geq 2^{\frac{s}{2} + 1 } + 5 \right]
    $
    and the binary cyclic code $\widehat{\mathcal{C}}_{(0,s)}$ has parameters
    $
    \left[ 2^s - 1, 2^{s-1} - 1 , d \geq 2^{\frac{s}{2}-1} - 1 \right].
    $
\item The dual code $\widehat{\mathcal{C}}_{(1,s)}^{\perp}$ has parameters
    $
    \left[ 2^s - 1, 2^{s-1} - 2 , d^\perp \geq 2^{\frac{s}{2}-1} \right]
    $
    and the dual code $\widehat{\mathcal{C}}_{(0,s)}^{\perp} $ has parameters
    $
    \left[ 2^s - 1, 2^{s-1}, d^\perp \geq 2^{\frac{s}{2}+1} + 6 \right].
    $
\end{enumerate}    
\end{theorem}

\begin{proof}
By 1) and 2) of Lemma \ref{lem-37}, the defining set  $\widehat{U}_{(1,s)}$ of the code  $\widehat{\mathcal{C}}_{(1,s)}$ contains the set $\{h : 1 \leq h \leq  2^{\frac{s}{2}+1} + 4\}$, while the defining set $\widehat{U}_{(0,s)}$ of  $\widehat{\mathcal{C}}_{(0,s)}$  contains the set $\{n - h : 1 \leq h \leq 2^{\frac{s}{2}-1} - 2\}$.
By the BCH bound, we derive
$
d(\widehat{\mathcal{C}}_{(1,s)}) \geq 2^{\frac{s}{2} + 1 } +5 \quad \text{and} \quad d(\widehat{\mathcal{C}}_{(0,s)}) \geq 2^{\frac{s}{2}} - 1.
$
Note that $\dim(\widehat{\mathcal{C}}_{(i,s)}) = n - |\widehat{U}_{(i,s)}|$ for each $i \in \{0,1\}$. The dimensions of the binary cyclic codes follow from Lemma \ref{lem-37}. Since $\dim(\widehat{\mathcal{C}}_{(i,s)}) = n - \dim(\widehat{\mathcal{C}}^{\perp}_{(i,s)})$, the dimensionss of the duals directly follow. The minimum distance of the duals are discussed in the following cases. 

Case 1: Let $i =1$. The set $\{h : 0 \leq h \leq 2^{\frac{s}{2}} - 2\} \subseteq \mathbb{Z}_{n} \setminus (-\widehat{U}_{(1,s)})$ by 2) of Lemma \ref{lem-37}.
 By the BCH bound,  the desired lower bound on  $d(\widehat{\mathcal{C}}^{\perp}_{(1,s)})$ is obtained.
 
Case 2: Let $i =0$. The set $\{n-h : 0 \leq h \leq 2^{\frac{s}{2}+1} + 4\} \subseteq \mathbb{Z}_{n} \setminus (-\widehat{U}_{(0,s)})$ by 1) of Lemma \ref{lem-37}.
 By the BCH bound,  the desired lower bound on  $d(\widehat{\mathcal{C}}^{\perp}_{(0,s)})$ is obtained. 
\end{proof}

The extended code of $\widehat{\mathcal{C}}_{(i,s)}$ are characterized as follows.

\begin{theorem}\label{th-39}
Let $s \equiv 2 \pmod{4}$ and $s \geq 10$. Then the extended code $\overline{\widehat{\mathcal{C}}_{(1,s)}}$ has parameters
    $$
    \left[ 2^s, 2^{s-1} + 1, d \geq 2^{\frac{s}{2} + 1 } + 6 \right]
    $$
    and the extended code $\overline{\widehat{\mathcal{C}}_{(0,s)}}$ has parameters
    $$
    \left[ 2^s, 2^{s-1} - 1, d \geq 2^{\frac{s}{2}-1}  \right].
    $$
\end{theorem}
\begin{proof}
 The desired results can be deduced from Theorem \ref{th-38} and the definition of the extended codes.
\end{proof}

\section{Summary and concluding remarks}
Let $W$ consist of unions of certain $2$-cyclotomic cosets modulo $n$ such that $0 \notin W$ and $W^c := \{1, 2, \cdots, n-1\} \setminus W$.
In this paper, by selecting some suitable sets $W$, we have presented six constructions of binary cyclic codes with defining set $W$ or $W^c$. The dimensions of the binary cyclic codes were characterized, and lower bounds on their minimum distances were derived.
The extended codes of the binary cyclic codes were also studied. The parameters of the binary cyclic codes and their duals have been listed in Table \ref{tab}. 
By Table \ref{tab},  these binary cyclic ones have rates near $1/2$, and simultaneously large minimal distances and large dual distances such
that the lower bounds of $d\cdot d^\perp$ are close to $n$ or $2n$. 
In particular, most of these binary codes improve on known parameters in Table \ref{tab}. In addition, our examples yield several binary codes with excellent parameters.

To the best of our knowledge, only a few known families of binary cyclic codes whose lower bounds of $d\cdot d^\perp$ reach $cn$ for $c\geq 1$, even without any restriction on the dimension $k$. Hence it is of interest to construct further infinite families of binary cyclic codes of length $n$ with large lower bounds on $d \cdot d^\perp$.

\end{document}